\documentclass[%
 reprint,
 amsmath,amssymb,
 aps,
pra,
]{revtex4-2}

\usepackage{graphicx}
\usepackage{dcolumn}
\usepackage{bm}
\usepackage[dvipsnames]{xcolor}
\usepackage{physics}
\usepackage{orcidlink}
\usepackage{mathrsfs}
\usepackage{hyperref}
\hypersetup{
    colorlinks = true,
    citecolor = magenta,
    linkcolor = magenta,
}

\usepackage{algorithm}
\usepackage{algpseudocode}
\usepackage{amsthm}
\usepackage{braket}
\usepackage[english]{babel}
\usepackage{enumitem}
\newtheorem{theorem}{Theorem}

\newtheorem{proposition}[theorem]{Proposition}
\newtheorem{lemma}[theorem]{Lemma}

\newtheorem{definition}[theorem]{Definition}

\newif\ifenablecommands  
\enablecommandsfalse

\begin{document}

\preprint{APS/123-QED}



\title{Stability of mixed-state phases under weak decoherence}

\author{Yifan F. Zhang}
\email{yz4281@princeton.edu}
\affiliation{Department of Electrical and Computer Engineering, Princeton University, Princeton, NJ 08544}
 
\author{Sarang Gopalakrishnan}
\email{sgopalakrishnan@princeton.edu}
\affiliation{Department of Electrical and Computer Engineering, Princeton University, Princeton, NJ 08544}
\date{\today}

\begin{abstract}

We prove that the Gibbs states of classical, and commuting-Pauli, Hamiltonians are stable under weak local decoherence: i.e., we show that the effect of the decoherence can be locally reversed. 
In particular, our conclusions apply to finite-temperature equilibrium critical points and ordered low-temperature phases. In these systems the \emph{unconditional} spatio-temporal correlations are long-range, and local (e.g., Metropolis) dynamics exhibits critical slowing down. Nevertheless, our results imply the existence of local ``decoders'' that undo the decoherence, when the decoherence strength is below a critical value. 
An implication of these results is that thermally stable quantum memories have a threshold against decoherence that remains nonzero as one approaches the critical temperature.
Analogously, in diffusion models, stability of data distributions implies the existence of computationally-efficent local denoisers in the late-time generation dynamics.

\end{abstract}

\maketitle

\section{Introduction }

In equilibrium statistical mechanics, phases are parameter regions in which the free energy evolves smoothly. This equilibrium perspective is well suited to to conventional solid-state experiments, in which the system of interest (e.g., the electron fluid in a metal) is well coupled to a heat bath. Present-day experiments in quantum devices necessitate a more general concept of phases: in these devices, systems are driven far from thermal equilibrium and are either isolated from the environment or coupled to engineered dissipative environments. A key step toward this general concept came from the development of quasi-adiabatic continuation~\cite{hastings2005quasiadiabatic}, for pure quantum states at zero temperature. According to this definition, phases are equivalence classes of quantum states such that two states in the same phase can be prepared from one another by an efficient process---specifically, a finite-depth local unitary (FDLU) circuit. This concept of pure-state phases (called FDLU-equivalence) reduces to the conventional one for ground states of gapped local Hamiltonians, but extends to \emph{any} quantum state, and connects naturally to questions in computational complexity thoery~\cite{anshu2020circuit,parham2025quantum,rosenthal2020bounds,nadimpalli2024pauli}. So far, this ``preparability'' perspective is only fully developed for pure quantum states and strictly unitary evolutions; thus, a natural task, which has seen intense recent activity, is its generalization to more general classes of mixed states and evolutions involving noise, measurement, and feedback~\cite{tantivasadakarn2023shortest,tantivasadakarn2023hierarchy,bravyi2022adaptive,smith2023deterministic,smith2024constant,sahay2024classifying,sahay2024finite,stephen2024preparing,zhang2024characterizing}. As an important special case, a classification of mixed states from the perspective of preparability would naturally extend to general classical probability distributions, of the type that routinely arise in machine learning, and that also seem to exhibit phase transitions~\cite{hu2025local}.

The most natural extension of FDLU-equivalence to mixed states would be to define two mixed states $\rho_1$ and $\rho_2$ as being in the same phase if there are local, finite-time physical processes under which $\rho_1 \mapsto \rho_2$ and $\rho_2 \mapsto \rho_1$. (In what follows will generally speak of these processes as quantum channels, or Lindblad master equations, as this terminology is more general; it encompasses classical Markov chains and Markov processes as a special case.) A crucial difference from the unitary setting is that the processes involved will generally be irreversible, so we will require the existence of \emph{separate} channels $\mathcal{N}, \mathcal{R}$ such that $\mathcal{N}\rho_1 \approx \rho_2, \mathcal{R}\rho_2 \approx \rho_1$. Thus this concept of equivalence is called ``two-way connectivity''~\cite{coser2019classification}. A canonical example of a mixed-state phase is a quantum error correcting code subject to noise levels that are below its threshold; in this example, $\mathcal{N}$ is the noise and $\mathcal{R}$ is the recovery operation that restores the initial quantum state. When the noise strength exceeds threshold, the recovery channel ceases to exist, and two-way connectivity fails; however, the \emph{noise} channel continues to be well-defined and local. Consequently, a key feature of mixed-state phase transitions like the error correction threshold is that they are not diagnosed by conventional correlation lengths, unlike equilibrium phase transitions: applying noise for a finite time can drive a system past the error correction threshold, but (by light-cone arguments~\cite{bravyi2006lieb}) cannot change the asymptotic behavior of any correlation function.

\begin{figure*}[t]
\includegraphics[width=0.8\linewidth]{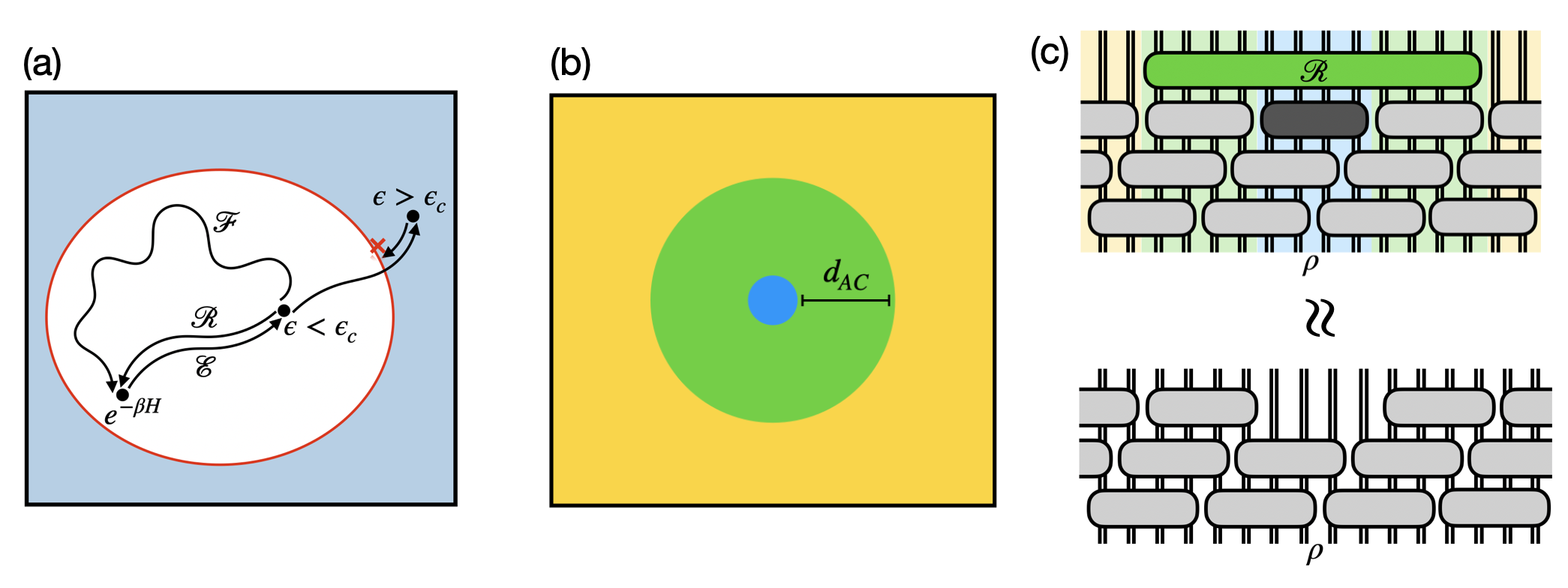}
\caption{\label{fig:phases_donut_recovery}(a) A phase diagram of a low-temperature or critical Gibbs state $e^{-\beta H}$ under local perturbations $\mathcal{E}$ with strength $\epsilon$. $\mathcal{R}$ is the recovery map we construct and $\mathcal{F}$ is the Gibbs sampler which could exhibit slowdown behavior. (b) An annulus-shaped tripartition $\textcolor{Cerulean}{A}\textcolor{Green}{B}\textcolor{Goldenrod}{C}$ used to define the Markov length. $\textcolor{Cerulean}{A}$ is a single qubit, $\textcolor{Green}{B}$ surrounds $\textcolor{Cerulean}{A}$ with a radius $d_{AC}$, and $\textcolor{Goldenrod}{C}$ is the rest of the system. (c) a local recovery channel $\mathcal{R}_{s,t}$ acting on $\textcolor{Cerulean}{A}\textcolor{Green}{B}$ that reverses the effect of $\mathcal{E}_{s,t}$ supported on $\textcolor{Cerulean}{A}$.}
\end{figure*}

Given the crucial part that correlation lengths play in conventional phase transitions, it is important to identify measures that play an analogous part at mixed-state phase transitions, for example, diverging at phase transitions and evolving smoothly inside a phase. A natural information-theoretic measure with these features was recently identified, namely the ``Markov length'' (which is defined in terms of conditional correlation functions; see Sec.~\ref{markov_def}). The Markov length has two key properties: first, it evades light-cone bounds~\cite{PRXQuantum.4.040332, zhang2024nonlocal, lee2024universal}, and can diverge at finite time; and second, as Ref.~\cite{sang2024stability} showed, any finite-time noisy local process that keeps the Markov length finite can be reversed by a finite-time evolution that is local on the scale of the Markov length. Thus, to show that a mixed state is stable against a class of perturbations, it suffices to show that the Markov length remains finite when these perturbations are sufficiently weak.  In the literature so far, this is assumed; however, as we discuss below, the validity of this assumption is not obvious, even at a physical level.

In the present work, we establish stability of the Markov length for a broad class of states---the Gibbs states of local, commuting-Pauli Hamiltonians---subject to local noise that obeys a mild technical condition, which is satisfied (e.g.) by depolarizing noise. Before the noise acts, these Gibbs states have zero Markov length by the Hammersley-Clifford theorem~\cite{brown2012quantum}. Our result states that the Markov length remains finite at least when the perturbation acts for a short enough time. We conjecture that this claim holds more generally for any Gibbs state of a local Hamiltonian at any nonzero temperature, subject to weak local noise: however, the techniques we use here do not extend to those cases.

We emphasize that the stability of the Markov length is not intuitively obvious. To see why, consider a system in equilibrium at a thermal critical point. Its Markov length is zero by virtue of being in equilibrium, but generic correlations are long-range, and one might think that any deviation from equilibrium would cause these to ``infect'' the conditional correlations that define the Markov length. Nor is it obvious that a perturbed critical point can be related to the equilibrium one by a short evolution. The simplest local dynamics (Metropolis or Glauber) that restores local equilibrium takes infinitely long to do so at the critical point, because of critical slowing down. Our result implies that, nevertheless, there exists a local stochastic process that relaxes fast to equilibrium. Essentially, the reason this process exists is because it is tailored to the noise model (unlike Metropolis dynamics). This is visualized in Fig.\ref{fig:phases_donut_recovery}(a). Thus, considered as part of a broader family of mixed states (or probability distributions), classical critical points are inside a phase rather than at the phase boundary. This is consistent with the fact that the equilibrium states immediately above and below the critical temperature lie in distinct mixed-state phases: going between these equilibrium states requires a long evolution during which the system is out of equilibrium and its Markov length can therefore diverge. An example of such a divergence was illustrated in Ref.~\cite{lloyd2025diverging}. 

\subsection{Comparison with earlier and concurrent works}\label{comparison}
We compare our results with earlier and concurrent works. The notion of Markov length in the context of mixed-state phases was introduced in Ref.~\cite{sang2024stability}, which proved that finite Markov length implies the existence of local recovery channels that reverse the effect of local noise. Subsequent works~\cite{sang2025mixed} use the finite Markov length condition as a definition of mixed-state phases. This definition is more refined than two-way connectivity: when $\rho_1$ evolves to $\rho_2$ long some path, the Markov length condition requires that $\rho_2$ can be evolved back to $\rho_1$ along the same path. This constrasts with two-way connectivity, which only requires the existence of some path from $\rho_2$ to $\rho_1$. An explicit example where two-way connectivity holds but the Markov length condition fails is given in Ref.~\cite{zhang2025conditional}.

 Subsequent work has attempted to establish the stability of the Markov length in various regimes, thereby showing the stability of mixed-state phases in different senses. Ref.~\cite{zhang2025conditional} shows that high-temperature commuting Gibbs states have a finite Markov length under arbitrarily strong but strictly local noise \footnote{Strictly speaking, Ref.~\cite{zhang2025conditional} only shows the finite Markov length under certain technical restrictions on the channel. These properties resemble the stabilzier mixing condition in this paper in the sense that it preserves the commutation property of the algebra genearting the Gibbs state. However, we believe this is a technical constraint; these results should hold for any high-temperature (non-commuting) Gibbs states under arbitrary strictly local channels.}. This implies that high-temperature Gibbs states are ``absolutely stable'' in the sense that no strictly local noise can drive a phase transition. 
More recently, Ref.~\cite{ma2025circuit} showed that if two Gibbs states (of potentially non-commuting local Hamiltonians) can be connected by interpolating their Hamiltonians while keeping a decay of symmetric correlations, then they are in the same mixed-state phase. This result shows that (under some nontrivial clustering assumptions, which would not hold in glassy phases) thermodynamic phases and mixed-state phases coincide for Gibbs states of local classical or quantum Hamiltonians.
%
This result also immediately implies the stability of mixed-state phases under weak Hamiltonian perturbations, if the Hamiltonian is deep in a gapped phase.

In addition, Ref.~\cite{yi2025universal} showed that for a family of gapped Hamiltonians, if the ground states of one of these Hamiltonians has a finite Markov length, then the ground states of all Hamiltonians in this family have a finite Markov length~\footnote{Technically speaking, their result only shows that if the ground states of one Hamiltonian has CMI decaying superpolynomially fast, then the ground states of all Hamiltonians in this family have CMI decaying superpolynomially fast. This is slightly weaker than a expoentially decay. The relaxation to superpolynomial decay also seems necessary because of the nature of quasi-adiabatic continuation. However, }. This can be understood as the generalization of FDLU stability of gapped ground states to the finite Markov length stability of mixed-state phases. They also generalize their results to mixed state phases assuming the local reversibility condition given in~\cite{sang2025mixed}.

Our result is complementary to Ref.~\cite{ma2025circuit,yi2025universal} in the sense that we consider the stability of mixed-state phases under local channels instead of Hamiltonian deformations. In particular, we do not assume the clustering of conventional correlations or a finite gap: thus, for example, our results continue to apply at thermal critical points.
%
For a critical Gibbs state, there exist arbitrarily weak Hamiltonian deformations that drive a phase transition (e.g., changing the temperature slightly across the critical temperature). However, systems take a long time to equilibrate across a phase transition, so such perturbations need to be on for a long time. On the other hand, our result shows that under weak local noise, even critical states lie inside a stable mixed-state phase. One can also understand our result as establishing the local reversibility condition in Ref.~\cite{sang2025mixed} for a broad class of Gibbs states under weak local noise.

\subsection{Organization}
The rest of this paper is structured as follows. In Sec.~\ref{context} we define the key concepts, state our result informally, and explore its implications. In the next three sections we prove our main result in three steps: first, for classical thermal states subject to single-site noise (Sec.~\ref{one_site}); then for classical thermal states subject to general local noise (Sec.~\ref{finite_depth}); and finally, the most general form, which applies to quantum commuting-Pauli Hamiltonians subject to noise that obeys a particular ``incoherence-preserving'' condition (Sec.~\ref{quantum_finite_depth}). Finally, we summarize our findings and comment on some open questions in Sec.~\ref{discussions}. 

\section{Context and main results}\label{context}

In this section we begin with a brief review of the current understanding of mixed-state phases and the Markov length. We then introduce some key definitions and informally state our main result. Finally, we discuss some immediate implications of this result for quantum error correction and the properties of diffusion models in machine learning. 

\subsection{Context: Markov length and mixed-state phases}\label{markov_def}

\emph{Gibbs states}.---We first define Gibbs (i.e., thermal) states. We consider systems of $q$-dimensional qudits, with one qudit on each of $n$ nodes of a hypergraph $\mathcal{G}$. We consider a Hamiltonian consisting of mutually commuting terms that live on hyperedges $a \in \mathcal{G}$, i.e., 
\begin{equation}\label{eq:hamiltonian}
H = \sum\nolimits_{a \in \mathcal{G}} h_a,
\end{equation}
For concreteness one could consider the Ising model in $d$ dimensions, where the sites are vertices of the graph and interaction terms live on each edge. Since the graphs that appear in machine-learning applications and quantum codes are quite general, it is helpful to state our results in the general case.

We will always consider \emph{commuting} Hamiltonians, i.e., Hamiltonians in which every pair of $h_a$, $h_b$ commute. As a sperical case of this, we will consider \emph{classical} Hamiltonians, in which every $h_a$ is a diagonal operator in the computational basis. The Gibbs state at inverse temperature $\beta$ is defined as
\begin{equation}
    \rho_\beta = Z^{-1} e^{-\beta H}, \quad Z = \mathrm{Tr}(e^{-\beta H}).
\end{equation}

Next, we introduce the Markov property and the Markov length. commuting Gibbs states are special in the sense that they are \emph{exactly Markov} at all temperatures.

\emph{Markov property}.---We first define the mutual information (MI) and the conditional mutual information (CMI). For any state $\rho$, the MI between two subsystems $A$ and $B$ is defined as
\begin{equation}
    I_{\rho}(A:B) \equiv S(\rho_A) + S(\rho_B) - S(\rho_{AB}),
\end{equation}
Where $\rho_L$ denotes the reduced density matrix on subsystem $L$, and $S(\rho) = -\mathrm{Tr}[\rho \log \rho]$ is the von Neumann entropy. The CMI for any three subsystems $A$, $B$, and $C$ is defined as
\begin{equation}
    I_{\rho}(A:C|B) \equiv I_{\rho}(A:BC) - I_{\rho}(A:B).
\end{equation}
MI and CMI are both well-defined for classical and quantum systems. MI measures the direct correlations between two subsystems. CMI is more subtle: it measures the correlations between $A$ and $C$ that are not mediated by $B$. We will always consider ``annular'' geometries of the form shown in Fig.~\ref{fig:phases_donut_recovery}(b), in which $B$ separates $A$ from the rest of the system, $C$. The minimal graph distance between a point in $A$ and one in $C$ is denoted $d_{AC}$. For classical probability distributions, $I(A:C|B)$ is the mutual information between $A$ and $C$, evaluated in the conditional distribution $P(AC|B=b)$, averaged over configurations $b$ of region $B$ with the appropriate weights:
\begin{equation}
    I_{P(ABC)}(A:C|B) = \sum_{b} P(b) I_{P(AC|B=b)}(A:C)
\end{equation}
Thus it is manifestly non-negative. In quantum systems, the relation $I(A:C|B) \geq 0$ is a way of stating the strong subadditivity of entropy~\cite{lieb1973proof}. In either the classical or the quantum case, a state satisfying the condition $I(A:C|B) = 0$ is said to be ``exactly Markov''~\footnote{The rationale for this terminology is that if we think of $A,B,C$ as respectively the past, present, and future, in a Markov chain the past affects the future only through its effect on the present.}. If a state $\rho_{ABC}$ is exactly Markov for some tripartition $ABC$, there is a quantum channel $\mathcal{R}: B \to AB$ such that $\mathcal{R}(\rho_{BC}) = \rho_{ABC}$~\cite{hayden2004structure}. Thus, $B$ holds all the information needed to recover $A$, regardless of the state of $C$. For classical distributions that are exactly Markov, the recovery map $\mathcal{R}$ can be constructed explicitly as the conditional probability distribution: $P(ABC)=P(A|B)P(BC)$.

In practice, many of the quantum states we consider will not be exactly Markov. Instead, if $I(A:C|B) \alt e^{-|B|/\xi}$ for all tripartitions with sufficiently large $|B|$, one says the state is approximately Markov. For approximately Markov states, a recovery map of the type discussed above can still be defined, but instead of recovering the state exactly it recovers it with a fidelity controlled by $e^{-|B|/\xi}$. Specifically, if $I(A:C|B) \leq \epsilon$, there is a recovery map $\mathcal{R}: B \to AB$ such that~\cite{fawzi2015quantum,junge2018universal}
\begin{equation}
    \Vert \rho_{ABC} - \mathcal{R}(\rho_{BC}) \Vert_1 \leq \sqrt{2 \epsilon}
\end{equation}

Following Ref.~\cite{sang2024stability}, we define the Markov length as the lengthscale over which the CMI decays exponentially.
\begin{definition}
    Consider the family of all possible tripartitions $ABC$ of the topology defined above (Fig.~\ref{fig:phases_donut_recovery}(b)). We say that a state $\rho$ has a Markov length $\xi$ if for any such tripartition, the CMI can be bounded as
    \begin{equation}
       I_{\rho}(A:C|B) \le c \, e^{-d_{AC}/\xi}
    \end{equation}
    Where $c$ is an constant that could depend polynomially on the tripartition geometry; $\xi$ is defined to be the Markov length.
\end{definition}
%
%

Commuting Gibbs states are exactly Markov at all temperatures for any tripartition $ABC$ such that $B$ separates $A$ from $C$, as a consequence of the Hammersley-Clifford theorem~\cite{brown2012quantum}. The intuition behind this result is straightforward in the classical case: the joint distribution $P(ABC)$ as be written as $\exp(-\beta H_{AB}) \exp(-\beta H_{BC})$, where $H_{AB}$ contains all terms in the Hamiltonian that act on $A$ and $B$, and similarly for $H_{BC}$. Conditioning on $B=b$ freezes all terms in $H_{AB}$ and $H_{BC}$ to some fixed value, so the conditional distribution factorizes as 
\begin{equation}
    P(AC|B=b) \propto \exp(-\beta H_{AB}(b)) \exp(-\beta H_{BC}(b)),
\end{equation}
which is uncorrelated between $A$ and $C$. A similar reasoning applies to commuting quantum Hamiltonians, although the ``conditioning'' operation is more subtle~\cite{hayden2004structure}. The above statement holds for any temperature. This behavior of CMI contrasts with the conventional MI which can be long-ranged at low temperatures and diverges at thermal critical points.

Non-commuting quantum Gibbs states are not generally exactly Markov; however, they often exhibit a finite Markov length, depending on the choice of tripartitions and the temperature~\cite{kuwahara2025clustering,kato2504clustering,chen2025quantum,bakshi2025dobrushin,bluhm2025strong,zhang2025conditional}. In particular, for the tripartition in Fig.~\ref{fig:phases_donut_recovery}(b), it is known that non-commuting quantum Gibbs states have a finite Markov length controlled by the temperature~\cite{chen2025quantum}\footnote{Strictly speaking, the result of Ref.~\cite{chen2025quantum} only shows the exponential decay of the recovery error with $d_{AC}$. When mapping this to the decay of CMI, one incurs an additional prefactor $|C|$ which renders the bound vacuous in the thermodynamic limit. However, recovery error is sufficeint for most purposes.}. The temperature-dependence of the Markov length is inevitable: in the zero-temperature limit the Gibbs state is just the ground state, which is a pure state, and for pure states $I(A:C|B) = I(A:C)$. This mutual information diverges for quantum-critical ground states, so there cannot be a uniform bound on the Markov length in this limit. Nevertheless, at any nonzero temperature---including at finite-temperature transitions with a ``quantum'' character, like the loss of a thermally stable quantum memory---Gibbs states are approximately Markov, even if their unconditional correlations are long-ranged.


\emph{Mixed-state phases}.---We define mixed-state phases in this section. To keep the discussion simple, we will specialize to systems that are geometrically local in Euclidean space. As we discussed in the introduction, two mixed states $\rho_1, \rho_2$ are said to be in the same mixed-state phase if there are \emph{admissible} channels $\mathcal{N}$ and $\mathcal{R}$ such that $\mathcal{N}\rho_1 = \rho_2, \mathcal{R}\rho_2 = \rho_1$. Under the refined definition of Ref.~\cite{coser2019classification}, a channel $\mathcal{N}$ is admissible if it can be written as a time-dependent evolution for a finite period $\tau$ under a (generally time-dependent) quasilocal Lindbladian, i.e., $\mathcal{N} = \mathcal{N}_{t = \tau}$, where for $0 \leq t \leq \tau$ $\mathcal{N}_t = \exp\left(\int_0^t \mathcal{L}(t') dt'\right)$, and moreover the state $\mathcal{N}_t \rho_1$ has a finite Markov length for all $t \leq \tau$. Here, a quasilocal Lindbladian is one in which any term is supported on at most $O(\mathrm{polylog} (n))$ qubits for a system of $n$ qudits. 

The key result relating mixed-state phases to the Markov length can be stated informally as follows. 

\begin{theorem}\label{shengqi}
(Theorem 1 of \cite{sang2024stability}) Suppose that $\mathcal{N}\rho_1 = \rho_2$ where $\mathcal{N}$ is an admissible channel as defined above, acting on a system of $n$ qubits, and that the Markov length of $\mathcal{N}\rho_t$ is upper-bounded by $\xi$ for all $t \leq \tau$. Then a recovery channel $\mathcal{R}$ consisting only of terms with radius $\leq r$ exists, such that 
\begin{equation}
\Vert (\mathcal{R} \circ \mathcal{N}) [\rho_1] - \rho_1 \Vert_1 \leq \mathrm{poly}(n) e^{-r/(2\xi)},
\end{equation}
where $\xi$ is the Markov length. In particular, if we choose $r \geq \xi \log \left(\mathrm{poly}(n)/\epsilon\right)$, we can achieve a recovery channel such that $\Vert (\mathcal{R} \circ \mathcal{N}) [\rho_1] - \rho_1 \Vert_1 \leq \epsilon$. 
\end{theorem}
We give an intuition behind the proof of this theorem. The first step is to Trotterize the Lindbladian into a finite-depth quantum channel consisting of layers of local channels. Then, one constructs a local recovery channel for each local channel in the trotterization, using the finite Markov length to ensure that the recovery channel can be localized to a neighborhood of size $\sim \xi \, \log (\rm{poly}(n)/\epsilon)$ (Fig.~\ref{fig:phases_donut_recovery}(c)). Composing these local recovery channels together gives the desired recovery channel $\mathcal{R}$. Thus the actions of an admissible channel on $\rho_1$ can be approximately inverted by another admissible channel consisting of terms of radius polylogarithmic in the system size $n$ as well as depth polylogarithmic in $n$. 
%
As a particular application of this result, an error correcting code like the toric code can be decoded one patch at a time, provided the patches are of radius $\agt \xi \log n$. Ref.~\cite{sang2024stability,yang2025topological} provide numerical evidence for the divergence of $\xi$ as one approaches a threshold. For toric code, this threshold coincides with the error correction threshold. 

We will use a slightly weaker (``Trotterized'') notion of what an admissible channel is: we will only require that the channel can be broken up into a large but finite number of time steps, and that the Markov length should remain finite at every time step. We note that the proof in Ref.~\cite{sang2024stability} applies whenever the Trotter step is sufficiently small. 

\subsection{Summary of main results}\label{main_results}

We now turn to the main results of the present work. We consider the Gibbs states of Hamiltonians of the form~\eqref{eq:hamiltonian}. We primarily consider commuting Hamiltonians in which each term $h_a$ is a product of $q$-dimensional Pauli operators. Speficially, the $q$-dimensional Pauli group is generated by the generalized $X$ and $Z$ operators defined as
\begin{align}
    X |j\rangle &= |j+1 \mod q \rangle \\
    Z |j\rangle &= \omega^j |j\rangle, \quad \omega = e^{2\pi i/q}
\end{align}
All classical Hamiltonians are special cases of this, in which all $h_a$ are products of $Z$ operators only. As an example, each $h_a$ could be an local element of the stabilizer group. Note that the set $\{h_a\}$ need not be a set of independent generators. For instance, The two-dimensional Ising model have terms that are not independent since the product of all closed loops of $ZZ$ terms is the identity operator.

We consider Gibbs states subject to arbitrary finite-depth graph-local channels. On a general graph these can be defined as follows: 

\begin{definition}\label{def:finite_depth_channel}
   Consider the interaction graph $\mathcal{G}$ defined in the previous section. We define a finite-depth quantum channel $\mathcal{E}$ with $T$ layers as follows 
\begin{align}\label{eq:finite_depth_channel}
    \mathcal{E} = \mathcal{E}_T \circ \mathcal{E}_{T-1} \circ \cdots \circ \mathcal{E}_1 \\
    \mathcal{E}_{t} = \prod_{a \in \mathcal{G}_t} \mathcal{E}_{a,t}
\end{align}
Where $\mathcal{E}_{t}$ denotes a layer of quantum channels, and each of $\mathcal{E}_{a,t}$ is a quantum channel acting on the qudits contained in the hyperedge $a$ at time $t$. $\mathcal{G}_t$ denotes a subgraph of $\mathcal{G}$ where every hyperedge is disconnected from each other, so that all $\mathcal{E}_{a,t}$ can be applied in parallel.
\end{definition}

Since we would like to show stability under weak perturbations, we will consider channels $\mathcal{E}_{a,t}$ that are close to the identity channel. Specifically, we assume that each local channel $\mathcal{E}_{a,t}$ can be written as
\begin{equation}\label{eq:perturbed_channel}
    \mathcal{E}_{a,t} = (1 - \epsilon_{a,t}) \mathcal{I} + \epsilon_{a,t} \mathcal{N}_{a,t},
\end{equation}
where $\mathcal{I}$ is the identity channel, $\mathcal{N}_{a,t}$ is an arbitrary quantum channel, and $\epsilon_{a,t} \in [0,1]$ quantifies the strength of the perturbation. We set an upper bound $\epsilon_{a,t} < \epsilon$ for all $a,t$.


We will primarily consider a restricted class of channels called \emph{stabilizer-mixing channels}, which we now give a informal introduction. The Hamiltonian~\eqref{eq:hamiltonian} can be explicitly diagonalized in terms of mixed stabilizer states specified by the eigenspaces of each stabilizer. Its Gibbs states can be completely characterized in terms of the classical probability mass over eigenspaces $P(\mathbf{s})$. We call it the \emph{stabilizer distribution} of the Gibbs state.

Our result will concern channels that have the following property:
\begin{definition}\label{def:stabilizer_mixing_channel_informal}
    Suppose the commuting Hamiltonian $H$ defined in Eq.~\eqref{eq:hamiltonian} can be decomposed into $\sum_{\mathbf{s}} E_{\mathbf{s}} \Pi_{\mathbf{s}}$, where $\Pi_{\mathbf{s}}$ is the projector onto the eigenspace labeled by $\mathbf{s}$ (proportional to the mixed stabilizer state). A channel $\mathcal{E}$ is called \emph{stabilizer-mixing} if for any $\mathbf{s}$ we have
\begin{equation}
    \mathcal{E}\left(\Pi_{\mathbf{s}}\right) = \sum_\mathbf{\mathbf{s}'} Q(\mathbf{s}') \Pi_{\mathbf{s}'},
\end{equation}
where $Q(\mathbf{s}')$ is a probability distribution over stabilizer states labelled by $\mathbf{s}'$.
\end{definition}

When the Hamiltonian is classical. we can always choose $\Pi_{\mathbf{s}}$ to be one-dimensional projectors onto computational basis states. In this case, stabilizer-mixing channels reduce to arbitrary stochastic processes.

Stabilizer-mixing channels give us two properties that are crucial for our analysis. First, density matrices that are diagonal in the eigenbasis of the Hamiltonian~\eqref{eq:hamiltonian} remain diagonal after the noise acts, i.e., the noise does not create quantum coherence in this basis. Second, the action of stabilizer-mixing channels preserve the degeneracy of each eigenspace. The two properties together ensures that the probability mass over the eigenspaces uniquely determines the density matrix after the noise acts.


We will introduce our notation more formally in the main text. For now we start by stating our main result:

\begin{theorem}\label{thm:quantum_stability_finite_depth_informal}
    (Informal) Consider a commuting Gibbs state $\rho \propto e^{-\beta H}$ where $H$ is defined in Eq.~(\ref{eq:hamiltonian}) and each $h_a$ is a product of Pauli operators. For all finite-depth local quantum channels $\mathcal{E}$ with $d$ layers of local channels defined in Definition \ref{def:finite_depth_channel} where each local channel $\mathcal{E}_{a,t}$ is (a) stabilizer-mixing as defined in Definition \ref{def:stabilizer_mixing_channel_informal} and (b) admits the form \eqref{eq:perturbed_channel}. There exists a constant $\epsilon_c$ depending on $\mathcal{G}$ and $d$ such that if $\epsilon_{a,t} < \epsilon_c$, then the Markov length of $\mathcal{E}(\rho)$ is finite.
\end{theorem}
Together with Theorem~\ref{shengqi}, this result implies the stability of commuting Gibbs states under local perturbations. We conjecture that the Markov length remains finite even without restricting to commuting-projector Hamiltonians and eigenspace-preserving channels, since under coarse-graining any Hamiltonian eventually flows to a zero-correlation-length limit. However, the techniques we have used to establish Theorem~\ref{thm:quantum_stability_finite_depth_informal} do not readily generalize to these cases. In any case, even the restricted result we are able to prove has broad implications, which we will discuss below.

The essential idea behind our proof is simplest to explain if we simplify the setting, and consider a classical Hamiltonian subject to single-site noise (Theorem~\ref{thm:classical_stability_single_site_informal} below). Establishing that the Markov length stays finite amounts to establishing that the \emph{classical} conditional distribution on $AC$ approximately factorizes given the state on $B$. By the data processing inequality, it suffices to consider noise acting only on $B$. Let us consider conditioning on a particular configuration of $B$, namely $b_i$. For this configuration, the Gibbs state (by hypothesis) factorizes perfectly between $A$ and $C$. Adding noise to $B$ is equivalent to having slightly imperfect knowledge of each spin in $B$ (so it has some probability $\sim \epsilon$ of pointing opposite to how it would in the configuration $b_i$). A key insight is that the conditional distribution on $ABC$ incorporating this imperfect knowledge can be represented as a Gibbs state of a modified Hamiltonian, in which a random field of strength $\sim |\log \epsilon|$ is applied to every spin in $B$, imperfectly pinning it along the configuration $b_i$. When $\epsilon$ is sufficiently small, this modified Gibbs state is in the trivial phase in region $B$, with exponentially decaying MI. This in turn implies the exponential decay of the CMI in the noisy distribution. 

The central technical contribution of the present work is to establish this conclusion formally using series-expansion techniques. Once this is done, the extension to channels acting on multiple sites, or to commuting-projector Gibbs states, follows directly from blocking the microscopic qudits into larger effective degrees of freedom. 

\subsection{Implications for thermally stable quantum memories}

Our result has important implications for the phase diagram of thermally stable quantum memories, such as the four-dimensional toric code~\cite{dennis2002topological,pastawski2011quantum,brown2016quantum}. We will discuss this example for concreteness but our considerations apply quite generally. Below a critical temperature $T_c > 0$, the 4D toric code has multiple Gibbs states in the thermodynamic limit, corresponding to the distinct ground-state logical sectors~\cite{dobrushin1968description,lanford1969observables}. Since the Hamiltonian of the 4D toric code consists of commuting local projectors, the Gibbs states are strictly Markov, even at $T_c$. On the other hand, if one starts with the 4D toric code at $T < T_c$ and adds depolarizing noise of strength $p$, the Markov length increases until it diverges at the error-correction threshold $p_c(T)$ for the optimal decoder. Since the Markov length is zero at the thermal transition and infinite at the noise-driven transition, these clearly belong to distinct universality classes.

\begin{figure}[tb]
\begin{center}
\includegraphics[width=0.45\textwidth]{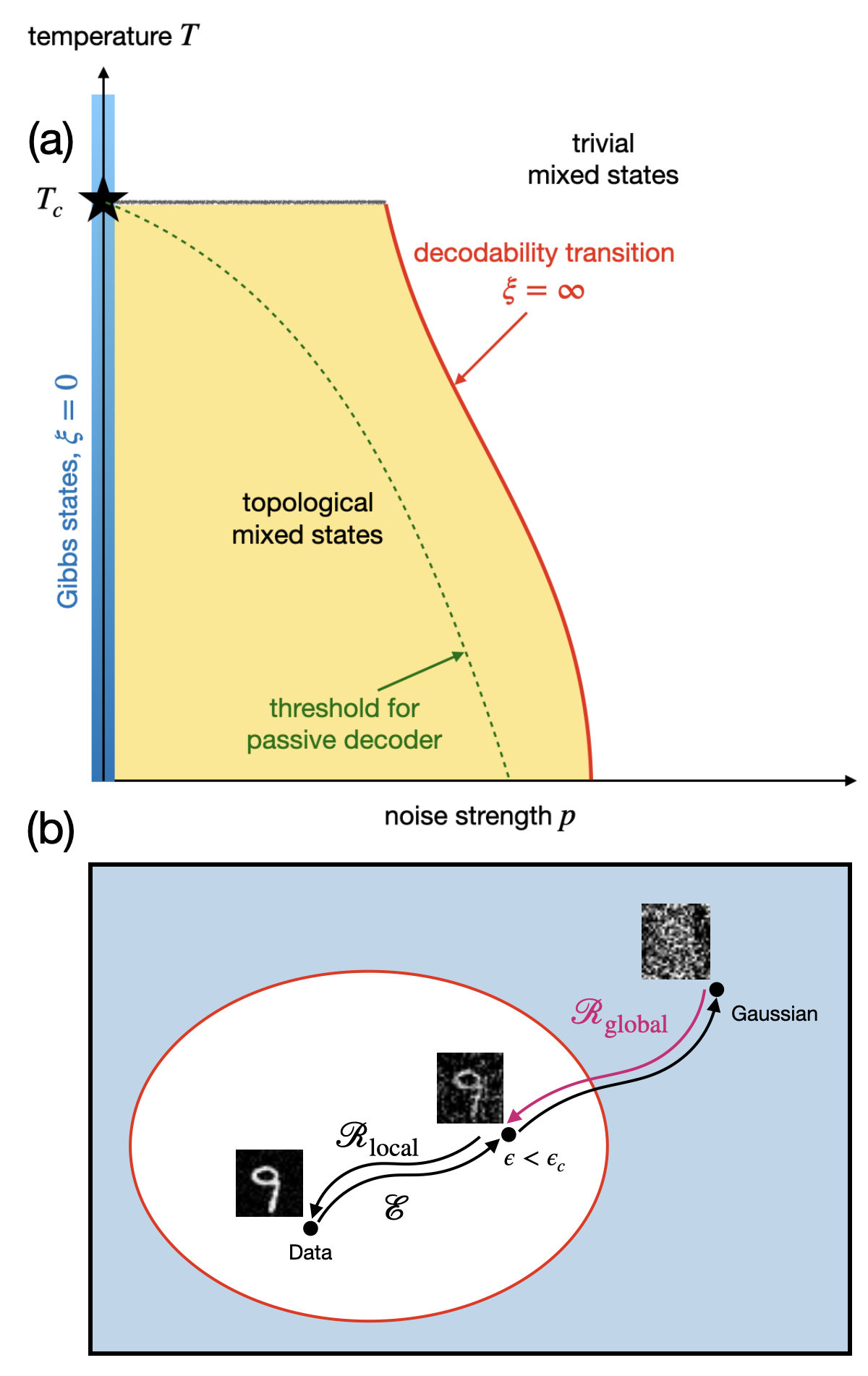}
\caption{(a) Schematic phase diagram for thermally stable quantum memories, as a function of temperature and noise strength. A mixed state with parameters $(T,p)$ is prepared by starting from a Gibbs state of temperature $T$ and applying noise of strength $p$ to every qudit. For noise below the information-theoretic threshold, quantum information can be recovered by optimal decoding. At the information-theoretic threshold, the Markov length diverges. Passive decoders such as heat-bath dynamics are suboptimal, so their thresholds have no information-theoretic significance. (b) Schematics of the diffusion model and its relation to mixed-state phases. The generation process has to violate locality at some point ($\mathcal{R}_{\rm{global}}$) but near the data distribution the generation dynamics can be local ($\mathcal{R}_{\rm{local}}$).}
\label{4dtc}
\end{center}
\end{figure}

A natural question is how these two phase transitions fit together as $T \to T_c$. Specifically, we consider a ``phase diagram'' as a function of temperature and noise strength, such that the state with parameters $(T,p)$ is prepared by initializing the system in a Gibbs state at temperature $T$ and applying a round of noise of strength $p$. Naively one might expect that $p_c(T) \to 0$ as $T \to T_c$. However, our result establishes that this \emph{cannot} be the case: even as $T \to T_c$, the Gibbs state remains perfectly Markov, and the noise can be reversed as long as its strength is below some $O(1)$ threshold. On the other hand, throughout the high-temperature phase, the Gibbs state contains no quantum information and there is nothing to decode. Thus the phase diagram of the optimal decoder must look like that sketched in Fig.~\ref{4dtc}(a), with a nonzero limiting value of the threshold as $T \to T_c^-$. The Gibbs states below and above $T_c$ are not related by a short-time evolution: if one tries to dynamically prepare a Gibbs state above $T_c$ by starting in a Gibbs state below $T_c$ and running local Metropolis dynamics, the system inevitably falls out of equilibrium and its Markov length diverges before falling back to zero~\cite{lloyd2025diverging}. 

So far, we have considered whether the information is \emph{in principle} retrievable using an optimal decoder. For a thermally stable memory, one could ask instead whether the information can be retrieved \emph{passively}, by coupling the system to a local heat bath. Passive decoders are suboptimal, and their thresholds lie strictly inside the phase of decodable states (Fig.~\ref{4dtc}). The threshold for passive decoding also vanishes as $T \to T_c^-$, because of the diverging susceptibility of the thermal state to perturbations. Thus the Markov length at the passive decoding transition goes smoothly to zero in this limit; however, since passive decoders are suboptimal, the value of the Markov length at the passive decoding threshold is a nonuniversal number that has no information-theoretic significance.

We emphasize that our current technical result does not rigorously imply a threshold theorem for quasilocal decoders at finite temperature, although it does imply a mixed-state phase diagram of the shape indicated. The precise result is that the true Gibbs state near $T_c$ (which is maximally mixed in the logical subspace) is in the same mixed-state phase as the decohered Gibbs state up to a finite decoherence strength. Moreover, our results yield an explicit quasi-local decoder that recovers the Gibbs state from the decohered Gibbs state. To prove a threshold, we would need to prove that this decoder recovers the Gibbs state with encoded logical information, starting from its decohered version. 

We argue for this informally following Ref.~\cite{sang2024stability}: the recovery channel is made up of gates that do not ``know'' about the logical state, since distinct logical states are locally indistinguishable at any $T < T_c$. Therefore, it will act the same way on the logical Gibbs state as it does on the true Gibbs state. 
Alternatively, since the logical Gibbs states become exactly Markov in the thermodynamic limit~\cite{dobrushin1968description,lanford1969observables}, applying our result to logical Gibbs states in the thermodynamic limit would suffice to establish a threshold. 
The main obstruction to doing this is that the existence of multiple Gibbs states below $T_c$ for the 4D toric code (corresponding to distinct logical sectors) has not been rigorously established to our knowledge. 
Our methods do not help prove this widely believed claim. Nevertheless, assuming the toric code has multiple distinct Gibbs states corresponding to logical sectors, our technical result implies a phase diagram of the form Fig.~\ref{4dtc}(a) even with logical information encoded. 
%
%
%
We leave a rigorous proof to future work.

\subsection{Local denoisers in diffusion models}

The concept of mixed-state phases and phase transitions is fundamentally buried in a class of generative AI model called the diffusion model~\cite{hyvarinen2005estimation,vincent2011connection,sohl2015deep,ho2020denoising,song2019generative,song2020improved,nichol2021improved}. There, the goal is to sample a distribution that we do not know but have samples from. Diffusion models achieve this by first adding noise to the data distribution through a local dynamics called the forward process, and then learning a reverse dynamics called the reverse process to recover the original data distribution (Fig.~\ref{4dtc}(b)). The reserve process that we wish to learn turns out to be the recovery map in mixed-state phases. One can also consider this process as a form of error correction, where the forward process is the noise channel and the reverse process is the decoder. However, instead of decoding a given state, diffusion model starts from the Gaussian noise and decodes it back to the data distribution. This does not fix the ``logical'' information, but it samples from the data distribution which is the goal of generative modeling.

While the forward process is local by construction (e.g. adding Gaussian noise to each pixel of an image independently), the reverse process is not guaranteed to be local. In fact, if the data distribution contains long-range correlations, which is common in real-world data, then the reverse process must be non-local since the data distribution is in a non-trivial phase. An important insight from Ref.~\cite{hu2025local} is that phase transition happens in a narrow time window. Away from this time window, the state is either deep in the trivial phase (at early time when the state is mostly noise) or deep in the data phase (at late time when the state is close to the data distribution). Since one expect the Markov length to be finite deep in the trivial/data phase, Theorem \ref{shengqi} implies that the reverse process can be well-approximated by a local dynamics away from the critical time window, potentially reducing the compute cost. This idea has seen success in Ref.~\cite{hu2025local}.

Our results provide a concrete example where local reverse processes provably exist in a constant time window. Specifically, we show that for any local Gibbs distribution, there exists a constant-strength noise such that the corrupted distribution can be locally recovered back to the original Gibbs distribution. This opens the door reducing the compute cost of diffusion models for data distributions, based on physics-inspired concepts like mixed-state phases and stability.

We also believe that real-life data distributions exhibit similar stability as local Gibbs distributions. For example, image data exhibits finite Markov lengths: to recover a lost pixel, one only need a small patch nearby instead of the entire image. While we do not have a proof of the stability of states at a finite Markov length, we believe that they are also stable under weak local perturbations, thereby ensuring the existence of local reverse processes in diffusion models for real-life data. We leave the proof of this conjecture to future work.

\section{Classical Gibbs States Under Weak Single-Site Stochastic Processes}\label{one_site}

\begin{figure*}
\includegraphics[width=\linewidth]{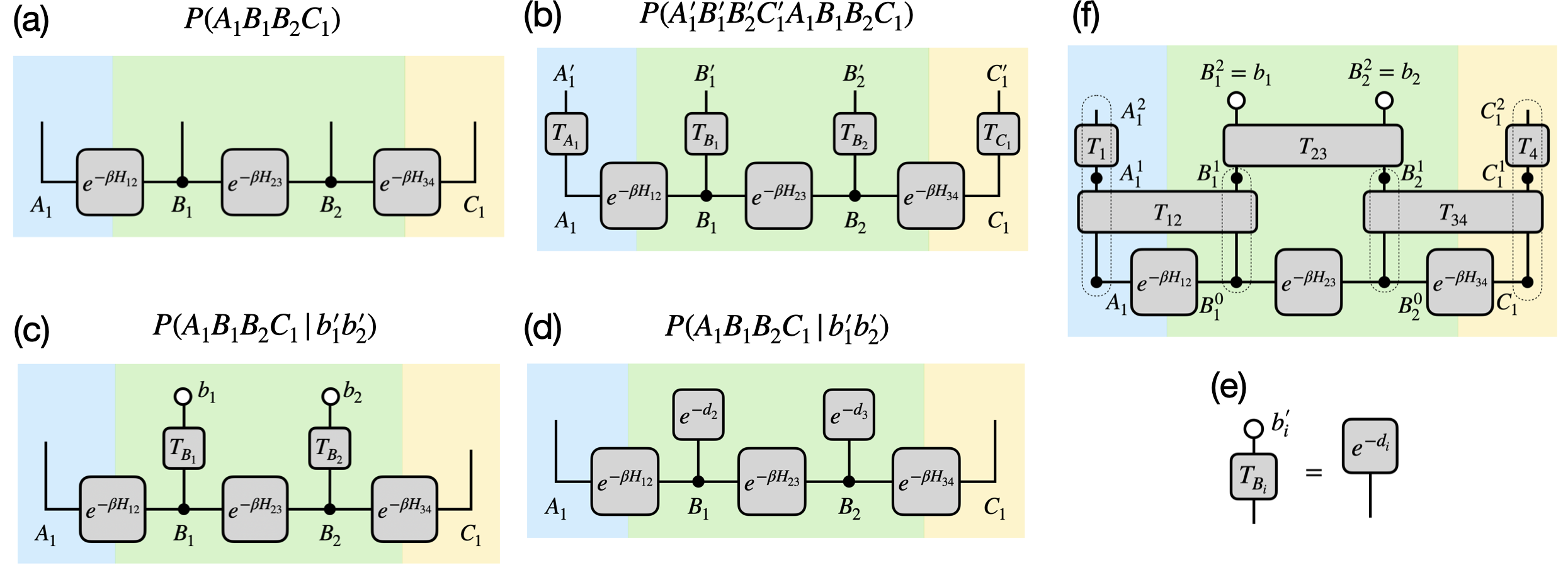}
\caption{\label{fig:pinning}(a) A one-dimensional Gibbs distribution $P(A_1 B_1 B_2 C_1)$ of four bits in a line and nearest-neighbor interactions. We partition the four bits into $\textcolor{Cerulean}{A_1}$, $\textcolor{Green}{B_1B_2}$, and $\textcolor{Goldenrod}{C_1}$.  (b) After applying a product of local stochastic processes on each bit, we obtain a joint distribution $P(A_1' B_1' B_2' C_1' A_1 B_1 B_2 C_1)$. (c) Conditioned on post-selecting $b_1' b_2'$ on $B$, the conditional distribution $P(A_1 B_1 B_2 C_1 | b_1' b_2')$. (d) The conditional distribution $P(A_1 B_1 B_2 C_1 | b_1' b_2')$ is a Gibbs distribution with pinning terms.  (e) Identifying the post-selected transition matrices as pinning terms. (f) blocking spins on the same site but in different time slices to form super-spins.}
\end{figure*}

We now turn to the proof of our main result. We will first establish this result in the special case of classical Gibbs states on which the noise acts independently on single qudits, and can be written as follows.
\begin{equation}\label{eq:independent_stochastic_process}
    \mathcal{T} = \bigotimes_{i=1}^n \mathcal{T}_i
\end{equation}
Where $\mathcal{T}_i$ is a stochastic matrix acting on the $i$-th spin. We represent the Gibbs states and regions $\textcolor{Cerulean}{A} \textcolor{Green}{B} \textcolor{Goldenrod}{C}$ in Fig.~\ref{fig:pinning}(a).




We first state our result below.

\begin{theorem}\label{thm:classical_stability_single_site_informal}
    (informal) Consider an interaction graph $\mathcal{G}$ supporting a Gibbs distribution $P(\mathbf{x}) \propto e^{-\beta H(\mathbf{x})}$ where $H(\mathbf{x})$ is defined in Eq.~(\ref{eq:hamiltonian}) and each $h_a$ is diagonal in the computational basis. There exist a constant $\epsilon_c$ depending on $\mathcal{G}$ and $\beta$ such that for all independent local stochastic processes $\mathcal{T}$ defined in Eq.~(\ref{eq:independent_stochastic_process}) subject to Eq.~(\ref{eq:perturbed_channel}), if $\epsilon_i < \epsilon_c$, then the Markov length of $\mathcal{T}(P)$ is finite.
\end{theorem}

We will illustrate the main ideas of the proof using the one-dimensional example in Fig.~\ref{fig:pinning}(a). We use $A_1'$, $B_1'$, $B_2'$, and $C_1'$ to denote the spins after applying the local stochastic processes in Fig. \ref{fig:pinning}. The joint distribution of all spins is given by
\begin{equation}
    \begin{split}
        P(A_1' B_1' B_2' C_1' A_1 B_1 B_2 C_1) = \\
        \mathcal{T}_{A_1}(A_1'|A_1) \mathcal{T}_{B_1}(B_1'|B_1) \mathcal{T}_{B_2}(B_2'|B_2) \mathcal{T}_{C_1}(C_1'|C_1) \\
        \times P(A_1 B_1 B_2 C_1)
    \end{split}
\end{equation}
Where $\mathcal{T}_{A_1}(A_1'|A_1)$ is the transition matrix of $\mathcal{T}_{A_1}$ and similarly for the other spins. We visualize this joint distribution in Fig.~\ref{fig:pinning}(b). Physically, we only have access to the noisy state, that is the marginal distribution $P(A_1' B_1' B_2' C_1')$. The joint distribution is only a mathematical tool that we will use to analyze the noisy state.

To bound the Markov length of $P(A'B'C')$, we first post-select the spins in $B'$ after applying the local stochastic processes. This gives us a conditional distribution $P(A'C'|b')$ where $b'$ is the post-selected configuration on $B'$. Then, the CMI of $P(A'B'C')$ can be written as the expected MI of $P(A'C'|b')$ over the post-selection $b'$. This is further upper bounded by the maximum MI of $P(A'C'|b')$ over $b'$.

\begin{proposition}\label{prop:post_selection}
    The CMI of $P(A'B'C')$ is equal to the MI of the conditional distribution $P(A'C'|b')$ averaged over the post-selection $b'$, which is further upper bounded by the maximum MI of $P(A'C'|b')$ over $b'$.
    \begin{align}
        &I_{P(A'B'C')}(A:C|B) \\
        =& \sum_{b'} P(b') I_{P(A'C'|b')}(A:C) \\
         \le& \max_{b'} I_{P(A'C'|b')}(A:C)
    \end{align}
\end{proposition}


Next, because of the data processing inequality, the MI of $P(A'C'|b')$ is further upper bounded by the MI of $P(AC|b')$.

\begin{proposition}\label{prop:classical_data_processing}
    The MI of $P(A'C'|b')$ is upper bounded by the MI of $P(AC|b')$.
    \begin{equation}
        I_{P(A'C'|b)}(A:C) \le I_{P(AC|b')}(A:C)
    \end{equation}
\end{proposition}

Now the problem is reduced to understanding the conditional distribution $P(AC|b')$. We will consider the bigger distribution $P(ACB|b')$, which can be understood as a ``backward inference'' problem. We are given the noisy configuration $b'$ on $B'$ and we want to infer the clean distribution of $A$, $B$, and $C$. We visualize this distribution in Fig.~\ref{fig:pinning}(c).

Our most important observation is that $P(ACB|b')$ is also a Gibbs distribution with the same Hamiltonian $H$ but with additional pinning fields on $B$.

\begin{lemma}[Pinning Lemma]
    Consider a classical Gibbs distribution $P(\mathbf{x})$ with Hamiltonian $H(\mathbf{x})$ defined in Eq.~(\ref{eq:hamiltonian}). After applying a product of local stochastic processes $\mathcal{T} = \bigotimes_{i=1}^n \mathcal{T}_i$, the conditional distribution $P(\mathbf{x}|b')$ after post-selecting $b'$ on the noisy spins $\mathbf{x}'$ is also a Gibbs distribution with the same Hamiltonian $H(\mathbf{x})$ but with additional pinning fields on the spins in $B$.
    \begin{equation}
        P(\mathbf{x}|b') \propto e^{-\beta H(\mathbf{x}) - \sum_{i \in B} p_i(x_i)}
    \end{equation}
    Where $p_i(x_i) = -\log \mathcal{T}_i(b_i'|x_i)$ is the pinning field on the $i$-th spin in $B$.
\end{lemma}

\begin{proof}
\label{lem:pinning}
    The proof follows from a direct application of Bayes' theorem. First, we realize that the ``forward'' conditional distribution $P(b'|x)$ is given by the transition matrices of the local stochastic processes.
    \begin{equation}
        P(b'|x) = \prod_{i=1}^n \mathcal{T}_i(b_i'|x_i)
    \end{equation}
    Then, by Bayes' theorem, we have
    \begin{align}
        P(x|b') \propto P(b'|x) P(x) \\
        \propto \prod_{i=1}^n \mathcal{T}_i(b_i'|x_i) e^{-\beta H(x)} \\
        \propto e^{-\beta H(x) - \sum_{i \in B} p_i(x_i)}
    \end{align}
    Where in the second line we plugged in the expression of $P(b'|x)$ and $P(x)$, and in the last line we defined the pinning fields $p_i(x_i) = \log \mathcal{T}_i(b_i'|x_i)$.
\end{proof}

The above lemma is visualized in Fig.~\ref{fig:pinning}(c-e). Fig. \ref{fig:pinning}(c) shows the joint distribution $P(A_1 B_1 B_2 C_1 | b_1' b_2')$ after post-selecting $b_1' b_2'$ on $B'$. Fig.~\ref{fig:pinning}(e) identifies the post-selected transition matrices as pinning fields on $B$. Finally, Fig.~\ref{fig:pinning}(d) shows that the conditional distribution $P(A_1 B_1 B_2 C_1 | b_1' b_2')$ is a Gibbs distribution with the same Hamiltonian but with additional pinning fields on $B$.

At this point, showing a finite Markov length reduces to showing that the MI of $P(AC|b')$ decays exponentially with $d_{AC}$. Because of the pinning Lemma, the problem reduces to showing the decay of MI in a Gibbs distribution with local pinning fields. In the noiseless limit, $b'$ and $b$ are identical, so the pinning fields become infinitely strong and pin all spins in $B$ to $b=b'$. As long as the noise is weak, the pinning fields are still strong and suppress long-range fluctuations. Therefore, we expect the correlations to decay exponentially with $d_{AC}$ as long as the noise is below an $O(1)$ threshold. The relation between the noise strength and the pinning field strength is formalized below.

\begin{proposition}\label{prop:pinning_field_strength_single_site}
    Consider a local stochastic process $\mathcal{T}_i$ acting on the $i$-th spin defined in Eq.~(\ref{eq:perturbed_channel}). The pinning field $p_i(x_i) = \log \mathcal{T}_i(b_i'|x_i)$ induced by post-selecting $b_i'$ on the noisy spin $i'$ satisfies
    \begin{equation}
        p_i(x_i) \begin{cases}
            \le -\log(1 - \epsilon_i), x_i = b_i' \\
            \ge -\log(\epsilon_i), x_i \ne b_i' 
        \end{cases}
    \end{equation}
    In particular, $p_i(x_i)$ has a energy gap of at least $\log((1-\epsilon_i)/\epsilon_i)$ between the favored configuration $x_i = b_i'$ and other configurations $x_i \ne b_i'$.
\end{proposition}

The above proposition follows trivially from Eq. \eqref{eq:perturbed_channel} and the fact that transition matrix elements are bounded by one. The pinning Lemma also has the following statistical interpretation. In the orginal Markov chain $A-B-C$, once $B$ is given, $A$ and $C$ are independent. We do not have direct access to $B$, instead we are given a noisy version $B'$. Because of the noise, $A$ and $C$ are no longer independent even after conditioning on $B'$. However, if the noise is weak, then $B'$ still contains a lot of information about $B$. Therefore, conditioning on $B'$, we expect $A$ and $C$ to be approximately independent on a large length scale. In fact, the toy model we consideed here is called the \emph{hidden Markov model} in machine learning.

\begin{lemma}\label{lem:decay_of_correlation_informal}
    (informal) Consider the pinned Gibbs distribution given in Proposition \ref{prop:pinning_field_strength_single_site}. If the noise strength $\epsilon$ of the local stochastic processes is below a constant threshold $\epsilon_c$, then the mutual information between $A$ and $C$ is bounded by
    \begin{equation}
        I(A:C) = O \left( \min(|\partial A|,|\partial C|) e^{-d_{AC}/\xi} \right)
    \end{equation}
    Where $|\partial A|$ and $|\partial C|$ are the boundary sizes of $A$ and $C$, and $\xi = O(1/(\log(\epsilon_c) - \log(\epsilon)))$ is a correlation length that upper bounds the Markov length.
\end{lemma}

The proof of the above lemma is the key technical part of this paper. We give an overview here and leave the details to Appendix~\ref{decay_proof}. The proof is based on the cluster expansion and the abstract polymer 
model~\cite{kotecky1986cluster,dobrushin1996estimates,friedli2017statistical}
which expresses the free energy as a converged sum over local objects called polymers. Here, we define polymers as a connected subset of spins in $B$ (with connectivity defined by $\mathcal{G}$). Because of the pinning fields, spins in $B$ are energetically more stable when they align with $b'$. Therefore, we treat the spin configurations that do not align with $b'$ as excitations and perform a series expansion. We show that (1) the series converges exponentially fast when $\epsilon < \epsilon_c$, and (2) only terms on the order at least $d_{AC}$ can contribute to the MI. This gives us the desired exponential decay of MI.

\begin{proof}
    (proof of Theorem \ref{thm:classical_stability_single_site_informal}) We first use Proposition \ref{prop:post_selection} and Proposition \ref{prop:classical_data_processing} to reduce the problem to showing the decay of MI in $P(AC|b')$. Then, using the pinning Proposition \ref{prop:pinning_field_strength_single_site}, we identify $P(AC|b')$ as a pinned Gibbs distribution. Finally, we use Lemma \ref{lem:decay_of_correlation_informal} to show the decay of MI in the pinned Gibbs distribution. We give the formal proof with coefficients worked out in Appendix~\ref{constants}.
\end{proof}

\section{Classical Gibbs States Under Weak Finite-Depth Local Stochastic Processes}\label{finite_depth}

In this section, we extend our previous result to classical Gibbs states under finite-depth local stochastic processes, defined below:
\begin{definition}\label{def:finite_depth_stochastic_process}
   We define a finite-depth quantum channel $\mathcal{T}$ is a finite-depth local quantum channel (Definition \ref{def:finite_depth_channel}) where each local channel is a stochastic process. Explicitly, we have
\begin{align}\label{eq:finite_depth_stochastic_process}
    \mathcal{T} &= \mathcal{T}_T \circ \mathcal{T}_{T-1} \circ \cdots \circ \mathcal{T}_1 \\
    \mathcal{T}_{t} &= \prod_{a \in \mathcal{G}_t} \mathcal{T}_{a,t}
\end{align}
Where we use $\mathcal{T}$ instead of $\mathcal{E}$ to denote stochastic processes. All notations are the same as in Definition \ref{def:finite_depth_channel}.
\end{definition}

This generalizes our previous result to the case where local stochastic processes can have overlapping support. We first state our main result.

\begin{theorem}\label{thm:classical_stability_finite_depth_informal}
    (Informal) Consider an interaction graph $\mathcal{G}$ supporting a Gibbs distribution $P(\mathbf{x}) \propto e^{-\beta H(\mathbf{x})}$ where $H(\mathbf{x})$ is defined in Eq.~(\ref{eq:hamiltonian}) and each $h_a$ is diagonal in the computational basis. For all finite-depth local stochastic processes $\mathcal{T}$ with $d$ layers of gates defined in Definition \ref{def:finite_depth_stochastic_process} subject to Eq.~(\ref{eq:perturbed_channel}), there exists a constant $\epsilon_c$ depending on $\mathcal{G}$ and $d$ such that if $\epsilon_{a,t} < \epsilon_c$, then the Markov length of $\mathcal{T}(P)$ is finite.
\end{theorem}

The proof of the above theorem follows from a reduction to the single-site case. We visualize a four-site Markov chain $A-B-C$ subject to a two-layer local stochastic process in Fig. \ref{fig:pinning}(f). We will use $\mathbf{x}_0$ to denote the original spins, $\mathbf{x}_i$ to denote the spins after applying the $i$ layer of local stochastic processes. Consider the joint distribution $P(\mathbf{x}_d,\ldots,\mathbf{x}_0)$ of all spins at any spacetime location, which is given by
\begin{equation}\label{eq:joint_distribution_finite_depth}
    P(\mathbf{x}_d,\ldots,\mathbf{x}_0) = \prod_{t=1}^d P(\mathbf{x}_t|\mathbf{x}_{t-1}) P(\mathbf{x}_0)
\end{equation}
Where $P(\mathbf{x}_t|\mathbf{x}_{t-1})$ is the transition matrix of the $t$-th layer of local stochastic processes, and $P(\mathbf{x}_0)$ is the original Gibbs distribution. Given the tripartition $ABC$ of the original spins $\mathbf{x}_0$, we can naturally extend this tripartition to all spins at any time step $t$ by defining $\mathbf{x}_{A,t}$, $\mathbf{x}_{B,t}$, and $\mathbf{x}_{C,t}$ to be the spins in $A$, $B$, and $C$ at time step $t$. We also define $\mathbf{x}_A = (\mathbf{x}_{A,0}, \mathbf{x}_{A,1}, \cdots, \mathbf{x}_{A,d})$ to be the collection of all spins in $A$ at all time steps. Similarly, we define $\mathbf{x}_B$ and $\mathbf{x}_C$.

We are interested in the Markov length of the marginal distribution $P(\mathbf{x}_d)$. Therefore, we post-select the spins in $\mathbf{x}_{B,d}$ and consider the conditional distribution on the rest of the spins. Let $\tilde{x}_{B} = (\tilde{x}_{B,0}, \tilde{x}_{B,1}, \cdots, \tilde{x}_{B,d-1})$ be the spins in $B$ at all time steps except the last one (since they are post-selected). We also define $\tilde{x} = (\mathbf{x}_A, \tilde{x}_B, \mathbf{x}_C)$ to be the collection of all spins except the post-selected spins in $\mathbf{x}_{B,d}$. Similar to Lemma \ref{lem:pinning}, we can show that the conditional distribution $P(\tilde{x}|\mathbf{x}_{B,d})$ is also a Gibbs distribution, but with an additional ``time-like'' dimension.

\begin{lemma}\label{lem:spacetime_gibbs_finite_depth}
    Consider a classical Gibbs distribution $P(\mathbf{x})$ with Hamiltonian $H(\mathbf{x})$ defined in Eq.~(\ref{eq:hamiltonian}). After applying a finite-depth local stochastic process $\mathcal{T}$ defined in Definition \ref{def:finite_depth_stochastic_process}, the conditional distribution $P(\tilde{x}|\mathbf{x}_{B,d})$ after post-selecting $\mathbf{x}_{B,d}$ on the noisy spins $\mathbf{x}_d$ is also a Gibbs distribution with the following Hamiltonian
    \begin{equation}
        \beta \tilde{H}(\tilde{x}) = \beta H(\mathbf{x}_0) + \sum_{a,t} p_{a,t} (\mathbf{x}_{a,t-1},\mathbf{x}_{a,t})
    \end{equation}
    Where $\mathbf{x}_{a,t}$ denotes the subset of random variables supported on $a$ at time $t$ of the channel $\mathcal{T}_{a,t}$. $p_{a,t}(\mathbf{x}_{a,t-1},\mathbf{x}_{a,t})$ is given by the element-wise logarithm of the transition matrices $\log \mathcal{T}_{a,t}(\mathbf{x}_{a,t-1},\mathbf{x}_{a,t})$ defined in Definition \ref{def:finite_depth_stochastic_process}.
\end{lemma}

\begin{proof}
    The proof is again a direct application of Bayes' rule. Starting from Eq.~(\ref{eq:joint_distribution_finite_depth}), we can write the conditional distribution as
\begin{align}
    &P(\tilde{x} | \mathbf{x}_{B,d}) \propto P(\mathbf{x}_{d}, \mathbf{x}_{d-1}, \ldots, \mathbf{x}_{1} | \mathbf{x}_0) P(\mathbf{x}_{0}) \\
    &\propto \left(\prod_{a \in \mathcal{G}} \mathcal{T}_{a,d}(\mathbf{x}_{a,d-1},\mathbf{x}_{a,d}) \right) \\
    &\times \left(\prod_{a \in \mathcal{G}} \mathcal{T}_{a,d-1}(\mathbf{x}_{a,d-2},\mathbf{x}_{a,d-1}) \right) \\
    \ldots &\times \left(\prod_{a \in \mathcal{G}} \mathcal{T}_{a,1}(\mathbf{x}_{a,0},\mathbf{x}_{a,1}) \right) e^{-\tilde{H(\mathbf{x}_0)}}
\end{align}
By Identifying $\mathcal{T}_{a,t}(\mathbf{x}_{a,t-1},\mathbf{x}_{a,t})$ to $e^{-\sum_{a,t} p_{a,t} (\mathbf{x}_{a,t-1},\mathbf{x}_{a,t})}$, we obtain the local Gibbs distribution with the additional ``time-like'' dimension.
\end{proof}

As an example, In Fig. \ref{fig:pinning}(f), the conditional distribution on the rest of the spins after post-selecting $\mathbf{x}_{B,2}$ becomes a two-dimensional Gibbs distribution. The boundary term at $t=0$ is given by the original Hamiltonian $H(\mathbf{x}_0)$, The coupling along the time direction is nearest-neighbor and given by the transition matrices. The transition matrices also induce a coupling along the spatial direction. For example, in Fig. \ref{fig:pinning}(f), the transition matrix $\mathcal{T}_{12}$ couples $A_1^0$, $B_1^0$, $A_1^1$, and $B_1^1$.

Next, we perform a blocking procedure to remove the time-like dimension. After that, we will be able to exploit the fact that local stochastic processes are weak to give a pinning argument. Consider spins at all time steps in $\tilde{x}$. We use $X_i$ to denote the collection of all spins at the $i$-th spatial location at all time steps. For sites in $A$ and $C$, $X_i$ contains $d+1$ spins, while for sites in $B$, $X_i$ contains $d$ spins. As an illustration, we showing the blocking as the dashed boxes in Fig. \ref{fig:pinning}(f).

Now we treat each $X_i$ as a single spin with $q^{d+1}$ or $q^d$ possible values. The conditional distribution $P(\tilde{x}|\mathbf{x}_{B,d})$ becomes a Gibbs distribution of the new spins $\mathbf{X} = (X_1, X_2, \cdots, X_n)$ with a new Hamiltonian $\tilde{H}(\mathbf{X})$. The new Hamiltonian does not contain a time-like dimension anymore. Furthermore, it inherits the same interaction graph $\mathcal{G}$ as the original Hamiltonian $H(\mathbf{x})$, since the local stochastic processes only couple spins that are close in $\mathcal{G}$. Now, we can show that $\tilde{H}(\mathbf{X})$ energetically favors $X_i$ in $B$ to align with $\mathbf{x}_{B,d}$ because of the weak local stochastic processes.

\begin{lemma}\label{lem:pinning_finite_depth}
    The distribution $P(\tilde{X}|\mathbf{x}_{B,d})$ is a Gibbs distribution with Hamiltonian $\tilde{H}(\mathbf{X})$ defined as follows.
    \begin{equation}
        \beta \tilde{H}(\mathbf{X}) = \sum_{a \in \mathcal{G}} \beta h_a(\mathbf{x}_{a,0}) + \sum_{a \in \mathcal{G}} p_{a}(\mathbf{X}_a)
    \end{equation}
    Where $\mathcal{G}_B$ is the set of hyperedges that contain spins in the boundary region $B$, and $p_a(\mathbf{X}_a)$ is defined as $\sum_t \log(\mathcal{T}_{a,t})$. Furthermore, if $a$ is completely c  if $p_a(\mathbf{X}_a)$ energetically favors all temporal spin to align with $\mathbf{x}_{B,d}$ as follows.
    \begin{equation}
        p_a(\mathbf{X}_a) \begin{cases}
            \le -d \log(1 - \epsilon), & \text{if } x_{i,t} = x_{i,d},\, \forall i \in a, t \\
            \ge -\log(\epsilon), & \text{otherwise}
        \end{cases}
    \end{equation}
\end{lemma}

\begin{proof}
    The decomposition of the Hamiltonian follows directly from Lemma \ref{lem:spacetime_gibbs_finite_depth} and the blocking procedure. It remains to show the pinning structure of $p_a(\mathbf{X}_a)$.
    We take each $\mathcal{T}_{a,t}$ and write down the decomposition
    \begin{equation}
        \begin{split}
        &\mathcal{T}_{a,t}(\mathbf{x}_{a,t-1},\mathbf{x}_{a,t}) = \\
        &(1-\epsilon_{a,t}) I(\mathbf{x}_{a,t-1},\mathbf{x}_{a,t}) + \epsilon_{a,t} \mathcal{N}_{a,t}(\mathbf{x}_{a,t-1},\mathbf{x}_{a,t})
        \end{split}
    \end{equation}
    Where $I(\mathbf{x}_{a,t})$ is the identity stochastic process and $\mathcal{N}_{a,t}(\mathbf{x}_{a,t})$ is an arbitrary stochastic process. Suppose $a$ contains sites in $B$. We fix $\mathbf{x}_{a,t}$ and bound its energy contribution as follows.
    \begin{equation}
        \begin{split}
        &-\log(\mathcal{T}_{a,t}(\mathbf{x}_{a,t-1},\mathbf{x}_{a,t})) \\
         &\begin{cases}
            \le -\log(1 - \epsilon_{a,t}), & \text{if } \mathbf{x}_{a,t-1} = \mathbf{x}_{a,t} \\
            \ge -\log(\epsilon_{a,t}), & \text{otherwise}
        \end{cases}
        \end{split}
    \end{equation}
    We define $p_a$ as $\sum_t \log(\mathcal{T}_{a,t})$. If for every $i \in a$ that is in $B$, $x_{i,t} = x_{i,d}$ for all $t$, then we have
    \begin{equation}
        p_a(\mathbf{X}_a) \le -\sum_{t=1}^{d} \log(1 - \epsilon_{a,t})
    \end{equation}
    If at least one $i$ is not aligned with $x_{i,d}$ for some $t$, then we have
    \begin{equation}
        p_a(\mathbf{X}_a) \ge -\log(\epsilon_{a,t})
    \end{equation}
Upper bounding $\epsilon_{a,t}$ by $\epsilon$ gives the desired result.
\end{proof}

Note that we have only considered the pinning effect of $p_a(\mathbf{X}_a)$ where $a$ is entirely supported in $B$. This is sufficient for our purpose. When $a$ is only partially supported in $B$, $p_a(\mathbf{X}_a)$ still pinned the spins in $B$.

To bound the Markov length, it suffices to show the decay of MI between $A$ and $C$ in the pinned Gibbs distribution $P(\mathbf{X}|\mathbf{x}_{B,d})$. The fact that $\mathbf{X}_A$ and $\mathbf{X}_C$ contain more spins than $\mathbf{x}_A$ and $\mathbf{x}_C$ does not matter since removing spins cannot increase MI. Since the pinning term $p_a(\mathbf{X}_a)$ is not strictly local, we need an improved version of Lemma \ref{lem:decay_of_correlation_informal}. However, the proof idea is the same. We present the improved version as Lemma \ref{lem:decay_of_correlation} in Appendix~\ref{decay_proof}. 


\begin{proof}
    (Proof of Theorem \ref{thm:classical_stability_finite_depth_informal}) We first use Proposition \ref{prop:post_selection} to reduce the problem of bounding the MI of $P(\mathbf{x}_{A,d}\mathbf{x}_{C,d}|\mathbf{x}_{B,d})$ over any choice of $\mathbf{x}_{B,d}$. Next, we use Proposition \ref{prop:classical_data_processing} to further reduce the problem to bounding the MI of $P(\mathbf{x}_{A,d}\mathbf{x}_{C,d}|\mathbf{x}_{B,d})$ the MI of $P(\mathbf{X}_A \mathbf{X}_C|\mathbf{x}_{B,d})$ over any choice of $\mathbf{x}_{B,d}$. This step cannot decrease MI because of data-processing inequality. Then, we use Lemma \ref{lem:spacetime_gibbs_finite_depth} and the blocking procedure to realize that $P(\mathbf{X}_A \mathbf{X}_C|\mathbf{x}_{B,d})$ is a Gibbs distribution with the same interaction graph. Lemma \ref{lem:pinning_finite_depth} shows that the blocked Hamiltonian has a pinning structure. Finally, we use Lemma \ref{lem:decay_of_correlation} to the decay of MI in $P(\mathbf{X}_A \mathbf{X}_C|\mathbf{x}_{B,d})$. This completes the proof.
\end{proof}

\section{Commuting Gibbs States Under Weak Finite-Depth Local Channels}\label{quantum_finite_depth}

In this section, we further extend our results to commuting Pauli Gibbs states under finite-depth local quantum channels as defined in Definition \ref{def:finite_depth_channel}. Here, we only obtain a partial result where we have to impose certain restrictions on the local channels. However, we note the our current result already covers many physically relevant setups such as finite-temperature stabilizer states under weak depolarization channels. We will start by introducing some notations and tools we will need. We then state the restrictions on the local channels. Finally, we state and prove our main stability result.

\subsection{Stabilizer distribution of commuting Pauli Gibbs states}

We start by introducing some notations. Consider a commuting Gibbs state $\rho \propto e^{-\beta H}$ where $H$ is defined in Eq.~(\ref{eq:hamiltonian}) and all $h_a$ are Pauli operators that commute with each other. There are two possible scenarios: (1) all $h_a$ are independent operators, i.e., there is no non-trivial product of $h_a$ that equals identity; (2) there are some dependent operators, i.e., there exists a non-trivial product of $h_a$ that equals identity. An example of the first scenario is the toric code Hamiltonian, while an example of the second scenario is the two-dimensional Ising model, since the product of all plaquette terms equals identity.

In the first scenario, we define a set $\{\tilde{h}_{a}\} = \{h_a\}$. In the second scenario, we can always find a maximal independent subset of $\{h_a\}$ and define $\{\tilde{h}_{a}\}$ to be this independent subset. For any classical Hamiltonian, we can always choose $\{\tilde{h}_{a}\} = \{Z_a\}$. $\{\tilde{h}_{a}\}$ is, up to constant multiplier, a set of stabilizer generators.

We will extract a set of (possibly mixed) stabilizer states from $\{\tilde{h}_{a}\}$. We first diagonalize each $\tilde{h}_{a}$.
\begin{equation}
    \tilde{h}_a = \sum_{s_a=0}^{q-1} \omega^{s_a} \Pi_{a,s_a}
\end{equation}
Where $\omega \propto e^{2\pi i / q}$ is the $q$-th root of unity and $\Pi_{a,s_a}$ are the projectors onto the eigenspaces of $\tilde{h}_a$. We take all $\tilde{h}_a$ and define a projector onto the joint eigenspaces as follows.
\begin{equation}
    \Pi_{\mathbf{s}} = \prod_a \Pi_{a,s_a}
\end{equation}
Where $\mathbf{s} = \{s_a\}$ denotes the set of all eigenspaces. Each $\Pi_{\mathbf{s}}$ defines a (possibly mixed) stabilizer state since it is the projector onto a stabilizer subspace. We also define $R = \rm{Tr}[\Pi_{\mathbf{s}}]$ to be the rank of each stabilizer state. Note that $R$ is independent of the choice of $\mathbf{s}$ since all $\tilde{h}_a$ are independent operators.

We first observe that any commuting Gibbs state $\rho$ can be expressed as a mixture of these stabilizer states.
\begin{proposition}\label{prop:stabilizer_decomposition_gibbs}
    The commuting Pauli Gibbs state $\rho$ can be expressed as a mixture of stabilizer states as follows:
\begin{align}
    \rho \propto e^{-\beta H} = \sum_{\mathbf{s}}  P(\mathbf{s}) \frac{\Pi_{\mathbf{s}}}{R}
\end{align}
Where $P(\mathbf{s})$ is the probability of the stabilizer $\mathbf{s}$ in the Gibbs,state $\rho$, which we call the \emph{stabilizer distribution}:
\begin{equation}
    P(\mathbf{s}) := \frac{\mathrm{Tr}[\Pi_{\mathbf{s}}\rho]}{\mathrm{Tr}[\rho]}
\end{equation}
\end{proposition}

In other words, all information about the commuting Gibbs state $\rho$ is contained in the stabilizers $\Pi_{\mathbf{s}}$ and the stabilizer distribution $P(\mathbf{s})$. Next, we show that the stabilizer distribution $P(\mathbf{s})$ is classical Gibbs distribution with locality inherited from the interaction graph $\mathcal{G}$ of $H$.
\begin{lemma}\label{lem:stabilizer_distribution_gibbs}
    The stabilizer distribution $P(\mathbf{s})$ of a commuting Gibbs state $\rho$ on an interaction graph $\mathcal{G}$ is a classical Gibbs distribution
    \begin{equation}
        P(\mathbf{s}) \propto e^{-\beta H_{\rm{stab}}(\mathbf{s})}
    \end{equation}
    with Hamiltonian $H_{\rm{stab}}(\mathbf{s})$ given by
    \begin{equation}
        H_{\rm{stab}}(\mathbf{s}) = \sum_{a \in \mathcal{G}} h_{a}(\mathbf{s}_{a})
    \end{equation}
    Where $\mathbf{s}_a$ denotes the values of the stabilizers overlapping with hyperedge $a$, and $h_{a}(\mathbf{s}_a)$ is defined as
    \begin{equation}
        h_{a}(\mathbf{s}_a) = \Tr\left[\Pi_{a,s_a} h_a\right]
    \end{equation}
\end{lemma}
\begin{proof}
    We start from the definition of $P(\mathbf{s})$:
    \begin{align}
        P(\mathbf{s}) \propto \mathrm{Tr}[\Pi_{\mathbf{s}} e^{-\beta H}] \\
        = \mathrm{Tr}\left[\prod_{a} \left(\prod_{s \in \mathbf{s}_a}\Pi_{a,s_a} e^{-\beta h_a}\right)\right]
    \end{align}
    Where in the second line we act each $e^{-\beta h_a}$ on projectors with overlapping support with it (taken from $s \in \mathbf{s}_a$). We also used the fact that $\Pi_{a,s_a}=\Pi_{a,s_a}^2$ to insert as many $\Pi_{a,s_a}$ as we want so that each $e^{-\beta h_a}$ is paired with all projectors overlapping with it. Each term in bracket evaluates to the identity matrix multiplied by the eigenvalue of $e^{-\beta h_a}$ in the eigenspace selected by $\Pi_{a,s_a}$. We will call this eigenvalue $e^{-\beta h_a(\mathbf{s}_a)}$. Therefore, we have
    \begin{equation}
        P(\mathbf{s}) \propto \prod_{a} e^{-\beta h_a(\mathbf{s}_a)}
    \end{equation}
\end{proof}

Therefore, the stabilizer distribution $P(\mathbf{s})$ is a classical Gibbs distribution, and thus is exactly Markov. We also note that when all $h_a$ are all independent operators and if we choose $\{\tilde{h}_a\} = \{h_a\}$, then $P(\mathbf{s})$ is a product distribution since each $h_a$ only depends on $s_a$. On the other hand, when there are dependent operators in $\{h_a\}$, then $P(\mathbf{s})$ can have non-trivial correlations, as in the case of the two-dimensional Ising model where $P(\mathbf{s})$ is exactly the oriniginal Gibbs distribution.

Note that $P(\mathbf{s})$ is technically defined on a different hypergraph $\mathcal{G}_s$. One can construct $\mathcal{G}_s$ from $\mathcal{G}$ by the following procedure: for each hyperedge $a$ in $\{\tilde{h}_a\}$, we create a vertex in $\mathcal{G}_s$; for each hyperedge $b$ in $\mathcal{G}$, we look for all $a$ in $\{\tilde{h}_a\}$ that shares support with $b$ and create a hyperedge in $\mathcal{G}_s$ connecting all such $a$. One can see that $\mathcal{G}_s$ inherits the locality structure from $\mathcal{G}$. For example, if two qubits are far apart in $\mathcal{G}$, then any stabilizers supported on these two qubits must also be far apart in $\mathcal{G}_s$.

\subsection{Stabilizer mixing channels}

Next, we define the restriction we impose on the local quantum channels, which we call \emph{stabilizer mixing channels}. The main intution behind this definition is that we want the notion of stabilizer distribution to be well-defined even after applying the local quantum channels. In other words, we do not allow the local quantum channels to create coherence between different stabilizer states.

\begin{definition}[Stabilizer mixing channels]\label{def:stabilizer_mixing_channel}
A channel $\mathcal{E}$ is called \emph{stabilizer mixing} if for any set of stabilizer states $\mathbf{s}$ we have
\begin{equation}
    \mathcal{E}\left(\Pi_{\mathbf{s}}\right) = \sum_\mathbf{\mathbf{s}'} Q(\mathbf{s}') \Pi_{\mathbf{s}'},
\end{equation}
where $Q(\mathbf{s}')$ is a probability distribution over stabilizer states $\mathbf{s}'$.
\end{definition}

In human words, a stabilizer mixing channel maps a stabilizer projector to a linear combination of stabilizer projectors. For the case of Toric codes, dephasing and depolarizing channels are stabilizer mixing, while amplitude damping is not. Also for classical Hamiltonians, any classical stochastic channel is a stabilizer mixing channel.

However, sometimes one can make a non-stabilizer mixing channel stabilizer mixing by adding additional dissipation. For example, if an amplitude damping channel resets to state $\ket{0}$, we can add an additional dephasing channel that resets the state to $\ket{1}$ with equal probability. As long as the initial channel is weak, the additional channel is still weak, so we would still expect stability.

With the stabilizer mixing channels, the dissipated Gibbs state after the channel is still uniquely determined by the stabilizer distribution. Furthermore, the new stabilizer distribution can be obtained from the original stabilizer distribution by a finite-depth local stochastic process.

\begin{proposition}
\label{lem:stabilizer_distribution_after_channel}
Let $\mathcal{E}$ be the finite-depth local channel given in Eq.~\ref{eq:finite_depth_channel}, where each $\mathcal{E}_{a,t}$ is a local stabilizer mixing channel. Then the dissipated Gibbs state $\rho' = \mathcal{E}(\rho)$ admits the following decomposition:
\begin{equation}
    \rho' = \sum_{\mathbf{s}'} P'(\mathbf{s}') \Pi_{\mathbf{s}'}
\end{equation}
Where $P'(\mathbf{s}')$ is a new stabilizer distribution that can be obtained from the original stabilizer distribution $P(\mathbf{s})$ by the following stochastic process:
\begin{equation}
    P'(\mathbf{s}') = \mathcal{T}[P(\mathbf{s})]
\end{equation}
Where $\mathcal{T}$ is a finite-depth local stochastic process defined in Definition \ref{def:finite_depth_stochastic_process} with the local channels $\mathcal{T}_{a,t}$ determined by the local channels $\mathcal{E}_{a,t}$ as follows:
\begin{equation}
    \mathcal{T}_{a,t}(\mathbf{s}'_{a,t-1},\mathbf{s}_{a,t}) = \frac{1}{R} \text{Tr}\left[\mathcal{E}_{a,t}(\Pi_{\mathbf{s}_{a,t-1}}) \Pi_{\mathbf{s}'_{a,t}}\right]
\end{equation}
Where $\mathbf{s}_{a,t}$ collects the stabilizer labels supported on the hyperedge $a$ at time $t$ and $\mathbf{s}'_{a,t-1}$ is similarly defined for time $t+1$.
\end{proposition}

In the case of Toric code, dephasing channels create a pair of adjacent anyons with the dephasing probability. This is a local stochastic process.

\subsection{Proof of main result}
With the stabilizer distribution and stabilizer-mixing channels, we are ready to prove our main result Theorem \ref{thm:quantum_stability_finite_depth_informal}. The proof follows from a reduction to the classical finite-depth local stochastic process case. We first show that stabilizer distribution after the channel $P'(\mathbf{s}')$ has a finite Markov length. This follows from our previous result on classical finite-depth local stochastic processes (Theorem \ref{thm:classical_stability_finite_depth_informal}). Then, we show that the finite Markov length of $P'(\mathbf{s}')$ implies the finite Markov length of $\mathcal{E}(\rho)$. This is formalized in the Lemma below.

\begin{lemma}
\label{lem:stabilizer_markov_implies_state_markov}
Consider the dissipated Gibbs state given in Lemma~\ref{lem:stabilizer_distribution_after_channel}. If $P'_{\mathbf{s}}$ has a Markov length $\xi_s$, then $\rho'$ has a Markov length $\xi = \xi_s$.
\end{lemma}

\begin{proof}
    We take any tripartition $A$, $B$, and $C$ of qubits. Let $\mathbf{s}_A$, $\mathbf{s}_B$, and $\mathbf{s}_C$ be the stabilizer labels supported entirely on $A$, $B$, and $C$ respectively. In other words, these stabilizer are not supported on multiple regions. We also define $\mathbf{s}_{\partial A}$ and $\mathbf{s}_{\partial C}$ to be the stabilizer labels that are supported on both $A$ and $B$, and both $B$ and $C$ respectively. The union of $\mathbf{s}_A$, $\mathbf{s}_B$, $\mathbf{s}_C$, $\mathbf{s}_{\partial A}$, and $\mathbf{s}_{\partial C}$ gives the full set of stabilizer labels $\mathbf{s}$.

    We consider the CMI of $P'(\mathbf{s})$ between $\mathbf{s}_A \mathbf{s}_{\partial A}$ and $\mathbf{s}_C \mathbf{s}_{\partial C}$ conditioned on $\mathbf{s}_B$, which we will call $I_{P'}(A \partial A : C \partial C | B)$. By assumption, since $P'(\mathbf{s})$ has a Markov length $\xi_s$, we have
    \begin{equation}
        I_{P'}(A \partial A : C \partial C | B) \leqq C' e^{-d_{s,AC}/\xi_s}
    \end{equation}
    Where $d_{s,AC}$ is the distance between $A$ and $C$ in the interaction graph $\mathcal{G}_s$ of $P'(\mathbf{s})$ and is related to $d_{AC}$ defined on $\mathcal{G}$ by a $O(1)$ different (which correspond to the boundary thickness). Therefore, we will absorb the $O(1)$ difference into the constant $C'$ and do not distinguish between $d_{s,AC}$ and $d_{AC}$ in the following.

    Next, we relate $I_{P'}(A \, \partial A : C \, \partial C | B)$ to the CMI of $\rho' = \mathcal{E}[\rho]$ between $A$ and $C$ conditioned on $B$, which we will call $I_{\rho'}(A:C|B)$. We start from the definition of $I_{\rho'}(A:C|B)$:
    \begin{equation}
        I_{\rho'}(A:C|B) = S(\rho'_{AB}) + S(\rho'_{BC}) - S(\rho'_{B}) - S(\rho'_{ABC})
    \end{equation}
    Where $\rho'_{AB}$, $\rho'_{BC}$, $\rho'_{B}$, and $\rho'_{ABC}$ are the reduced density matrices of $\rho'$ on regions $AB$, $BC$, $B$, and $ABC$ respectively. We expand the reduced density matrices in terms of the stabilizer projectors as follows:
    \begin{equation}
        \rho'_{AB} = \frac{1}{R} \sum_{\mathbf{s}} P'(\mathbf{s}) \Tr_{C}[\Pi_{\mathbf{s}}]
    \end{equation}
    Next, notice that if taking the partial trace over $C$ washes out the stabilizer labels in $\mathbf{s}_C$ and $\mathbf{s}_{\partial C}$. Specifically, if $\mathbf{s}$ and $\mathbf{s}'$ are two stabilizer labels that only differ in $\mathbf{s}_C$ and $\mathbf{s}_{\partial C}$, then we have
    \begin{equation}
        \Tr_{C}[\Pi_{\mathbf{s}}] = \Tr_{C}[\Pi_{\mathbf{s}'}]
    \end{equation}
    Therefore, we denote $\Pi_{\mathbf{s}_{AB}} = \Tr_{C}[\Pi_{\mathbf{s}}]$ to be the reduced stabilizer projector on $AB$ after tracing out $C$. We can rewrite $\rho'_{AB}$ as follows:
    \begin{equation}
        \rho'_{AB} = \frac{1}{R} \sum_{\mathbf{s}_A, \mathbf{s}_{\partial A}, \mathbf{s}_B} P'(\mathbf{s}_A, \mathbf{s}_{\partial A}, \mathbf{s}_B) \Pi_{\mathbf{s}_{AB}}
    \end{equation}
    Where $P'(\mathbf{s}_A, \mathbf{s}_{\partial A}, \mathbf{s}_B)$ is the marginal distribution of $P'(\mathbf{s})$ on $\mathbf{s}_A$, $\mathbf{s}_{\partial A}$, and $\mathbf{s}_B$. Since different $\Pi_{\mathbf{s}_{AB}}$ are orthogonal projectors, we can compute the entropy of $\rho'_{AB}$ with the Holevo formula as follows:
\begin{align}
    S(\rho'_{AB}) 
    &= H(P'(\mathbf{s}_A, \mathbf{s}_{\partial A}, \mathbf{s}_B)) \notag\\
    &\quad + \sum_{\mathbf{s}_A, \mathbf{s}_{\partial A}, \mathbf{s}_B}
    P'(\mathbf{s}_A, \mathbf{s}_{\partial A}, \mathbf{s}_B)
    S\!\left(\frac{\Pi_{\mathbf{s}_{AB}}}{R}\right) \\[3pt]
    &= H\!\left(P'(\mathbf{s}_A, \mathbf{s}_{\partial A}, \mathbf{s}_B)\right)
    + S\!\left(\frac{\Pi_{\mathbf{s}_{AB}}}{R}\right)
\end{align}
    Where $H(P'(\mathbf{s}_A, \mathbf{s}_{\partial A}, \mathbf{s}_B))$ is the Shannon entropy of the marginal distribution $P'(\mathbf{s}_A, \mathbf{s}_{\partial A}, \mathbf{s}_B)$. In the second line, we used the fact that $S(\Pi_{\mathbf{s}_{AB}}/R)$ is independent of the choice of $\mathbf{s}_A$, $\mathbf{s}_{\partial A}$, and $\mathbf{s}_B$ since all stabilizer projectors have the same rank. Similarly, we can expand $S(\rho'_{BC})$, $S(\rho'_{B})$, and $S(\rho'_{ABC})$ as follows:
\begin{align}
    S(\rho'_{BC}) &= H\left(P'(\mathbf{s}_C, \mathbf{s}_{\partial C}, \mathbf{s}_B)\right) + S\left(\frac{\Pi_{\mathbf{s}_{BC}}}{R}\right) \\
    S(\rho'_{B}) &= H\left(P'(\mathbf{s}_B)\right) + S\left(\frac{\Pi_{\mathbf{s}_{B}}}{R}\right) \\
    S(\rho'_{ABC}) &= H\left(P'(\mathbf{s})\right) + S\left(\frac{\Pi_{\mathbf{s}}}{R}\right)
\end{align}
    Where $\Pi_{\mathbf{s}_{BC}}$ and $\Pi_{\mathbf{s}_{B}}$ are similarly defined reduced stabilizer projectors. We plug in the above four equations into the definition of $I_{\rho'}(A:C|B)$ and obtain
    \begin{align}
        &I_{\rho'}(A:C|B)
        = I_{P'}(A \, \partial A : C \, \partial C | B) \notag\\
        &\quad + S\left(\frac{\Pi_{\mathbf{s}_{AB}}}{R}\right) + S\left(\frac{\Pi_{\mathbf{s}_{BC}}}{R}\right) - S\left(\frac{\Pi_{\mathbf{s}_{B}}}{R}\right) - S\left(\frac{\Pi_{\mathbf{s}}}{R}\right)
    \end{align}
    Lastly, we recognize that the second line is the CMI of the stabilizer states between $A$ and $C$ conditioned on $B$, which is always zero since stabilizer states are exactly Markov. Therefore, we have
    \begin{equation}
         I_{\rho'}(A:C|B) = I_{P'}(A \partial A : C \partial C | B)
    \end{equation}
    This shows $\xi = \xi_s$.

\end{proof}

As a direct consequence of the above lemma, if $P'(\mathbf{s}')$ has a finite Markov length, then $\rho'$ is also bounded by the same Markov length.

\begin{proof}
    (Proof of Theorem \ref{thm:quantum_stability_finite_depth_informal}) We first use Lemma \ref{lem:stabilizer_distribution_after_channel} and Theorem \ref{thm:classical_stability_finite_depth_informal} to show that the eigenspace distribution $P'(\mathbf{s}')$ after the channel has a finite Markov length. Then, we use Lemma \ref{lem:stabilizer_markov_implies_state_markov} to show that the finite Markov length of $P'(\mathbf{s}')$ implies the finite Markov length of $\rho'$.
\end{proof}

\section{Discussions}\label{discussions}

We have shown that classical Gibbs states and commuting Pauli Gibbs states are stable under weak finite-depth local channels that are stabilizer mixing. Our results do not rely on the thermal correlation length being finite, so they apply to classical critical and ordered phases. Our result has particularly strong implications for stability around classical critical points. Even perturbations that act for a short time near a classical critical point lead to critical slowing down under local Metropolis or Glauber dynamics. However, our results imply that there are local channels that rapidly undo the effects of such perturbations. These apparently contradictory results are in fact consistent because the local recovery map we construct carries information about the noise model, while Metropolis and Glauber dynamics are not noise-aware. 

Our result opens up several interesting directions for future research. First, it would be interesting to extend our result to generic states with finite Markov length, for example, Gibbs states of non-commuting Hamiltonians. It seems plausible that generic states with finite Markov length are stable under weak finite-depth local channels, but our current techniques do not apply. In particular, the ``conditioning'' operation becomes ambiguous for approximate quantum Markov chains. Therefore, we expect that significant technical advances are needed to tackle this problem.

The second direction is to understand mixed-state phases beyond thermal Gibbs states, and universal properties of non-equilibrium mixed-state phases. We expect that techniques developed in \cite{sang2025mixed} together with our result (which proves their local reversibility assumption) can be helpful. In particular, we have shown the existence of new ``critical phases'' that correspond to only one point in the thermodynamic phase diagram. A natural question is to understand the property of these critical phases, in relation to the critical Gibbs states. 
The phase transition between different mixed-state phases is also an interesting open problem. In particular, under local channels, a local Gibbs state can evolved into another state whose parent Hamiltonian is non-local. This is already observed in the study of renormalization group transformation of the infinite Gibbs measures~\cite{van1991renormalization,van1993regularity}. In modern language, this corresponds to a breakdown of exact Markov property under local channels, and a finite Markov length indicates that the state is still ``local'' in some sense. However, at the mixed-state phase transition, the Markov length diverges, and it seems that the divergence of Markov length manifests a breakdown of locality in the presence of noise. An intriguing question is to understand the nature of this ``non-local critical point'' in mixed-state phase transitions.

The stability of mixed-state phases at zero temperature is another open problem. There has been numerical evidence that zero-temperature fixed-point mixed states are stable under weak local noise \cite{sang2024stability,yang2025topological}. However, a rigorous proof is still lacking. We envision that our current techniques can be adapted to tackle this problem by taking the zero-temperature limit of a commuting projector Hamiltonian carefully.

Lastly, it would be interesting to extend our result to show a threshold theorem for quasi-local decoders. In particular, one needs to handle the fact that code states having a logical observable destroys the exact Markov property. Because of the intuition from thermodynamic limit~\cite{dobrushin1968description,lanford1969observables}, we expect that the deviation from exact Markov property decays with system size. If one can quantify this deviation, then it seems plausible to extend our current techniques to show a threshold theorem for quasi-local decoders under stablizer-mixing noise channels.

\begin{acknowledgments}
The authors would like to thank Tim Hsieh, Vedika Khemani, Isaac Kim, Ethan Lake, Jong Yeon Lee, Ruochen Ma, Tibor Rakovszky, Daniel Ranard, and Shengqi Sang for useful discussions. Y.F.Z. and S.G. acknowledge support from NSF QuSEC-TAQS OSI 2326767.
\end{acknowledgments}


\begin{widetext}
\appendix

\section{Decay of Correlation Above a Critical Pinning Strength}\label{decay_proof}
In this section, we prove that correlation decays above a critical pinning strength for classical Gibbs states on hypergraph $G$ where each vertex is included in at most $\mathfrak{d}$ hyperedges, and each hyperedge acts on at most $\mathfrak{k}$ spins. We first define the Hamiltonian considered here. Let $\mathcal{G}_B$ be the subgraph of $\mathcal{G}$ where all hyperedges in $\mathcal{G}_B$ contain at least one site in $B$. We consider the most general form of Hamiltonian where the pinning term can be acting on any hyperedge in $\mathcal{G}$, not just on single sites. The Hamiltonian is defined as follows.
\begin{align}\label{eq:pinned_hamiltonian}
  H(\mathbf{x})
  &=
  \sum_{a\in\mathcal G} h_a\bigl(\mathbf{x}_{a}\bigr)
  \;+\;
  \sum_{a \in \mathcal{G}} p_a(x_{a}),
  \\
  h_a \ge 0, \qquad
  \lVert h_a\rVert_\infty &\le h_{\max},
  \qquad
  p_a(x_{a}) 
     \begin{cases}
       = 0,&x_{a}=0,\\[2pt]
       \ge p_{\min},&\text{otherwise}.
     \end{cases}, \, \text{when} \, a \, \text{is entirely supported on} B\nonumber
\end{align}
The hamiltonian consists of two parts: the first part is the interaction energy, and the second part is the pinning energy that penalises spins not being zero. $h_a \ge 0$ sets the minimal energy to at least zero. Note that the pinning term only acts on hyperedges supported on $B$. The interaction strength is bounded by $h_{\max}$. The pinning strength $p_{\min}$ is controlled by $\epsilon$.

We also define the distance between two regions $A$ and $C$ in $\mathcal{G}$ as follows.
\begin{definition}\label{def:distance_in_hypergraph}
    The distance $d_{AC}$ between two regions $A$ and $C$ in $\mathcal{G}$ is defined as the length of the shortest path between any site in $A$ to any site in $C$ in the interaction graph $\mathcal{G}$, where each pair of consecutive sites in the path must be connected by a hyperedge.
\end{definition}

We will always consider the tripartition $A$, $B$, and $C$ such that $ABC$ together form the entire system and (2) $B$ separates $A$ and $C$ in $\mathcal{G}$.
\begin{definition}
    A region $B$ separates regions $A$ and $C$ in $\mathcal{G}$ if any path from any site in $A$ to any site in $C$ must pass through $B$.
\end{definition}

The mutual information between $A$ and $C$ should be controlled by their boundary size, in the spirit of area law. We define the boundary set as follows.
\begin{definition}
    The boundary set $\partial A$ of region $A$ is defined as the set of sites in $A$ that are connected to sites outside of $A$ by at least one hyperedge in $\mathcal{G}$.
\end{definition}

We prove that there exists a constant $p_{\min,c}$ depending on $\mathcal{G}$ and $h_{\max}$ such that if $p_{\min} > p_{\min,c}$, then the connected correlation between any two observables $O_A$ and $O_C$ decays exponentially with the distance between $A$ and $C$ for any choice of $O_A$ and $O_C$. We state the result more formally below.
\begin{lemma} \label{lem:decay_of_correlation} 
   Consider a Gibbs distribution $P(\mathbf{x}) \propto e^{-H(\mathbf{x})}$ with $H(\mathbf{x})$ defined in Eq.~(\ref{eq:pinned_hamiltonian}). There exists a constant $p_{\min,c} = \mathfrak{d}(\delta_c + \mathfrak{d} h_{\max} + \log(q^d-1))$ such that if $p_{\min} > p_{\min,c}$, then the mutual information decays exponentially with the distance between $A$ and $C$:
   \begin{equation}
       I_P(A:C) \leq c \min(|\partial A|, |\partial C|) e^{-d_{AC}/\xi}
   \end{equation}
   Where $d_{AC}$ is the distance between $A$ and $C$ in $\mathcal{G}$. $|\partial A|$ is the number of sites on the boundary of $A$. $| \partial C|$ is similarly defined. $c$ is a constant, and $\xi$ is the Markov length.
\end{lemma}

\subsection{Outline of the Proof}\label{sec:proof_outline}
We first outline the main steps of the proof. We write down the following \emph{un-normalized} distribution on $AC$ by summing over all spins in $B$:
\begin{equation}\label{eq:tilde_P_definition}
    \tilde{P}(\mathbf{x}_{AC}) = \sum_{\mathbf{x}_B} e^{-\sum_{a \in \mathcal{G}} h_a(\mathbf{x}_a)+ p_a(x_a)}
\end{equation}

The main technical step is to show that $\tilde{P}(\mathbf{x}_{AC})$ factorizes as follows:
\begin{equation}\label{eq:tilde_P_factorization}
    \tilde{P}(\mathbf{x}_{AC}) = Z_{A}(\mathbf{x}_A) Z_{C}(\mathbf{x}_C) \exp (F_{AC}(\mathbf{x}_A, \mathbf{x}_C))
\end{equation}
Where $Z_{A}(\mathbf{x}_A)$ depends only on $\mathbf{x}_A$, $Z_{C}(\mathbf{x}_C)$ depends only on $\mathbf{x}_C$, and $F_{AC}(\mathbf{x}_{AC})$ is a function that depends on both $\mathbf{x}_A$ and $\mathbf{x}_C$, but $|F_{AC}(\mathbf{x}_{AC})|$ decays exponentially with the distance between $A$ and $C$. We will dedicate the rest of this section to prove this factorization and the decay of $F_{AC}(\mathbf{x}_{AC})$.

For now we assume the above factorization and show how it implies the decay of CMI. Let $\bar{Z}_A = \sum_{\mathbf{x}_A} Z_A(\mathbf{x}_A)$ and $\bar{Z}_C = \sum_{\mathbf{x}_C} Z_C(\mathbf{x}_C)$ be the normalization constants. We define two normalized marginal distributions $Q(\mathbf{x}_A)$ and $Q(\mathbf{x}_C)$ as follows:
\begin{align}
    Q_A(\mathbf{x}_A) &= Z_A(\mathbf{x}_A) / \bar{Z}_A \\
    Q_C(\mathbf{x}_C) &= Z_C(\mathbf{x}_C) / \bar{Z}_C
\end{align}

We now normalize $\tilde{P}(\mathbf{x}_{AC})$ to obtain $P(\mathbf{x}_{AC})$. We define the normalization constant $\bar{P}$ as follows:
\begin{align}
    \bar{P} &= \sum_{\mathbf{x}_{AC}} \tilde{P}(\mathbf{x}_{AC}) \\
    &= \bar{Z}_A \bar{Z}_C \sum_{\mathbf{x}_{A}} \sum_{\mathbf{x}_{C}} Q_A(\mathbf{x}_A) Q_C(\mathbf{x}_C) \exp(F_{AC}(\mathbf{x}_A, \mathbf{x}_C))
\end{align}
We recognize the last term as taking the average value of $\exp(F_{AC}(\mathbf{x}_A, \mathbf{x}_C))$ over the product distribution $Q_A(\mathbf{x}_A) Q_C(\mathbf{x}_C)$.Suppose for any choice of $\mathbf{x}_{AC}$, $|F_{AC}(\mathbf{x}_{AC})| \leq \epsilon$, where $\epsilon$ is a small number that decays exponentially with $d_{AC}$. Then, we have
\begin{equation}
    |\log(\bar{P}) - \log(\bar{Z}_A) - \log(\bar{Z}_C)| \leq \epsilon
\end{equation}

now we can write down the normalized distribution $P(\mathbf{x}_{AC})$ as follows:
\begin{equation}
    P(\mathbf{x}_{AC}) = \frac{\tilde{P}(\mathbf{x}_{AC})}{\bar{P}} = \frac{\bar{Z_A} \bar{Z}_C}{\bar{P}} \left( Q_A(\mathbf{x}_A) Q_C(\mathbf{x}_C) \exp(F_{AC}(\mathbf{x}_{AC})) \right)
\end{equation}
We can compare $\log(P(\mathbf{x}_{AC}))$ to $\log(Q_A(\mathbf{x}_A) Q_C(\mathbf{x}_C))$ as follows:
\begin{equation}
    |\log(P(\mathbf{x}_{AC})) - \log(Q_A(\mathbf{x}_A) Q_C(\mathbf{x}_C))| \leq \left| \log\left(\frac{\bar{Z}_A \bar{Z}_C}{\bar{P}}\right) \right| + |F_{AC}(\mathbf{x}_{AC})| \leq 2\epsilon
\end{equation}
This shows that $P(\mathbf{x}_{AC})$ is close to the product distribution $Q_A(\mathbf{x}_A) Q_C(\mathbf{x}_C)$ with an \emph{additive} error $2\epsilon$ in the log distribution, implying a \emph{multiplicative} error in the distribution itself.

Finally, we can use the above bound to show the decay of mutual information.
\begin{equation}
    I_P(A:C) =S(P_{AC} \| P_A P_C) \leq S(P_{AC} \| Q_A Q_C)
\end{equation}
Where $S(P||Q)$ is the relative entropy between distributions $P$ and $Q$, defined As $S(P||Q) = \sum_x P(x) \log(P(x)/Q(x))$. We use $P_{AC}$ as a short-handed notation of $P(\mathbf{x}_{AC})$. $Q_A$ and $Q_C$ are similar short-handed notations. $P_A$ and $P_C$ denote the marginal distributions on $A$ and $C$, respectively. $S(\cdot \| \cdot)$ is the relative entropy. The above inequality follows from the fact that the actuall marginals $P_A$ and $P_C$ minimize the relative entropy. We explicitly evaluate the relative entropy
\begin{align}
    S(P_{AC} \| Q_A Q_C) &= \sum_{\mathbf{x}_{AC}} P_{AC} \left( \log(P_{AC}) - \log(Q_A Q_C) \right) \le 2 \epsilon
\end{align}
Since $\epsilon$ decays exponentially with $d_{AC}$, we have shown the decay of mutual information.

One can see that the main technical step is to prove the factorization in Eq.~(\ref{eq:tilde_P_factorization}) and the decay of $F_{AC}(\mathbf{x}_{AC})$. We dedicate the rest of this section to prove this factorization.

\subsection{Preliminaries}
We start by rewriting the unnoramlized distribution $\tilde{P}(\mathbf{x}_{AC})$ in Eq.~(\ref{eq:tilde_P_definition}) as follows:
\begin{equation}
    \tilde{P}(\mathbf{x}_{AC}) = Z_0 \sum_{D \subset B} Z_D
\end{equation}
Where $D$ runs over all subsets of $B$. $Z_0$ is defined as the partition function but setting all sites in $B$ to zero:
\begin{equation}
    Z_0 = e^{-\sum_{a \in \mathcal{G}} h_a(\mathbf{x}_a)+ p_a(\mathbf{x}_a)}  \Bigr|_{\mathbf{x}_B = 0}
\end{equation}
And $Z_D$ is defined as the partition function with sites in $D \subset B$ unpinned:
\begin{equation}
    Z_D = \frac{1}{Z_0} \sum_{\substack{\mathbf{x}_D \\ x_i \neq 0, \forall x_i \in \mathbf{x}_D}} e^{-\sum_{a\in\mathcal{G}} h_a(\mathbf{x}_a)+ p_a(\mathbf{x}_a)}
    \end{equation}
Where $\mathbf{x}_D$ runs over all configurations on $D$ where none of the sites are set to zero. Note that when $D$ is empty, $Z_D = 1$. Also, $Z_0$ and $Z_D$ always depend on $\mathbf{x}_{AC}$, and we omit this dependence in the notation for simplicity.

We first note that in the geometry where $B$ separates $A$ and $C$, $Z_0$ factorizes as follows:
\begin{proposition}
\label{prop:Z_0_factorization}
In the geometry where $B$ separates $A$ and $C$ in the sense that one cannot find a path of hyperedges connecting $A$ and $C$, $Z_0$ factorizes as follows:
\begin{equation}
    Z_0 = Z_{0,A} Z_{0,C}
\end{equation}
Where $Z_{0,A}$ only depends on $\mathbf{x}_A$ and $Z_{0,C}$ only depends on $\mathbf{x}_C$.
\end{proposition}
\begin{proof}
    We break the Hamiltonian into two parts: $H_{AB}$ that only acts on $A$ and $B$, $H_{BC}$ that only acts on $B$ and $C$. This is possible because $B$ separates $A$ and $C$. Therefore, we have
    \begin{equation}
        Z_0 = e^{-H_{AB}(\mathbf{x}_{AB})} e^{-H_{BC}(\mathbf{x}_{BC})} \Bigr|_{\mathbf{x}_B = 0} = Z_{0,A} Z_{0,C}
    \end{equation}
    Where $Z_{0,A} = e^{-H_{AB}(\mathbf{x}_{AB})} \Bigr|_{\mathbf{x}_B = 0} $ and $Z_{0,C} = e^{-H_{BC}(\mathbf{x}_{BC})} \Bigr|_{\mathbf{x}_B = 0}$.
\end{proof}

Next, we make two observations about $Z_D$. The first observation is that $|Z_D|$ decays exponentially with $|D|$ when $p_{\min}$ is above a critical threshold.

\begin{lemma}
\label{lem:Z_D_bound}
For any $D \subset B$ and for any $\mathbf{x}_{AC}$, $Z_D$ is bounded as follows:
\begin{equation}
    |Z_D| \leq e^{-\delta |D|}
\end{equation}
Where $\delta = p_{\min}/\mathfrak{d} - \mathfrak{d} h_{\max} - \log(q-1)$.
\end{lemma}
\begin{proof}
    We initialize initialize $\mathbf{x}_B$ to all zeros and then flip all sites in $D$ to any non-zero values. The pinning term penalies an energy by at least $|D| p_{\min} / \mathfrak{d}$, since flipping $|D|$ sites must trigger at least $|D|/\mathfrak{d}$ pinning terms $p_a$. The interaction term can contribute at most $\mathfrak{d} |D| h_{\max}$ to the energy since each site belongs to at most $d$ hyperedges. Therefore, the total energy cost is at least $|D| p_{\min}/\mathfrak{d} - \mathfrak{d} |D| h_{\max}$. Let $\delta E = p_{\min}/\mathfrak{d} - \mathfrak{d} h_{\max}$. Therefore, for any $\mathbf{x}_{AC}$ we have
    \begin{equation}
\frac{
  \displaystyle
  \sum_{\substack{\mathbf{x}_D \\ x_i \neq 0\ \forall i\in D}}
  e^{-\sum_{a\in\mathcal{G}} h_a(\mathbf{x}_a)+ p_a(\mathbf{x}_a)}
}{
  \displaystyle
  e^{-\sum_{a\in\mathcal{G}} h_a(\mathbf{x}_a)+ p_a(\mathbf{x}_a)}\Bigr|_{\mathbf{x}_D = 0}
}
\;\le\;
(q-1)^{|D|}\, \exp\!\big(-|D|\,(p_{\min}/\mathfrak{d} -\mathfrak{d}\,h_{\max})\big).
\end{equation}
Therefore, we have
\begin{equation}
    |Z_D| \leq (q-1)^{|D|} e^{-|D| (p_{\min}/\mathfrak{d}  - \mathfrak{d} h_{\max})} \leq e^{-|D| (p_{\min}/\mathfrak{d} - \mathfrak{d} h_{\max} - \log(q-1))}
\end{equation}
\end{proof}

Next, we observe that when $D$ is formed by disconnected components, then $Z_D$ factorizes.
\begin{lemma}
\label{lem:Z_D_factorization}
If $D$ can be decomposed into disconnected components $D = D_1 \cup D_2$ where $D_1$ and $D_2$ induce two disconnected subgraphs in $\mathcal{G}$, then $Z_D$ factorizes as follows:
\begin{equation}
    Z_D = Z_{D_1} Z_{D_2}
\end{equation}
\end{lemma}
\begin{proof}
    Let $H_D$ be the Hamiltonian components that act on $D$.We first rewrite $Z_D$ in terms of a Hamiltonian that only acts on $D$ and its neighborhood:
    \begin{equation}
        Z_D = \frac{\displaystyle \sum_{\substack{\mathbf{x}_{D} \\ x_{i}\neq 0\ \forall i\in D}}
    \displaystyle  e^{-\sum_{a \, \text{supp}\, D} h_{a}(\mathbf{x}_a) + p_{a}(\mathbf{x}_{a})}}{
    \displaystyle e^{-\sum_{a \, \text{supp}\, D}h_{a}(\mathbf{x}_a) + p_{a}(\mathbf{x}_{a})}\Bigr|_{\mathbf{x}_{D}=0}}
    \end{equation}
    Where $\mathbf{x}_{\bar{D}}$ is the collection of spins seen by $H_D$ which includes $\mathbf{x}_D$ and its neighborhood. Note that $\mathbf{x}_{\bar{D}}$ may contain spins in $AC$. $a \, \text{supp} \, D$ means that $a$ contains spins in $D$, and we sum over all such $a$. In human words, we have cancelled out all terms that do not depend on $\mathbf{x}_D$ in the fraction.

    Given that $D$ factorizes into two disconnected components $D_1$ and $D_2$, we can decompose $H_D$ into two parts that act on $D_1$ and $D_2$ separately: $H_D = H_{D_1} + H_{D_2}$. Correspondingly, we can decompose $\mathbf{x}_{\bar{D}}$ into two parts $\mathbf{x}_{\bar{D}_1}$ and $\mathbf{x}_{\bar{D}_2}$, where $\mathbf{x}_{\bar{D}_1}$ is the collection of spins seen by $H_{D_1}$ and $\mathbf{x}_{\bar{D}_2}$ is similarly defined. Therefore, we have
    \begin{equation}
        Z_D = \frac{\left(\displaystyle \sum_{\substack{\mathbf{x}_{D_{1}} \\ x_{i}\neq 0\ \forall i\in D_{1}}}
    \displaystyle  e^{-\sum_{a \, \text{supp}\, D_{1}} h_{a}(\mathbf{x}_a) + p_{a}(\mathbf{x}_{a})}\right)
    \left(\displaystyle \sum_{\substack{\mathbf{x}_{D_{2}} \\ x_{i}\neq 0\ \forall i\in D_{2}}}
    \displaystyle  e^{-\sum_{a \, \text{supp}\, D_{2}} h_{a}(\mathbf{x}_a) + p_{a}(\mathbf{x}_{a})}\right)}{\left(\displaystyle e^{-\sum_{a \, \text{supp}\, D_{1}}h_{a}(\mathbf{x}_a) + p_{a}(\mathbf{x}_{a})}\Bigr|_{\mathbf{x}_{D_{1}}=0}\right)
    \left(\displaystyle e^{-\sum_{a \, \text{supp}\, D_{2}}h_{a}(\mathbf{x}_a) + p_{a}(\mathbf{x}_{a})}\Bigr|_{\mathbf{x}_{D_{2}}=0}\right)}
    \end{equation}
    One can see that the above equation is exactly $Z_{D_1} Z_{D_2}$, which proves the lemma.
\end{proof}

We will exploit the exponential decay of $Z_D$ and the factorization property of $Z_D$ to write down a polymer expansion of $\tilde{P}(\mathbf{x}_{AC})$.
\subsection{Polymer Expansion}
We introduce the polymer expansion formally in this section. We first define a polymer as a connected component of $D$.

\begin{definition}[Polymers and compatibility]
A \emph{polymer} \(\gamma\subseteq\{1,\dots,n\}\) is a connected
subset of sites with respect to the interaction graph \(\mathcal G\).
Its \emph{weight}, denoted as $|\gamma|$, is the number of sites in \(\gamma\).
Two polymers are \emph{compatible} (\(\gamma\sim\gamma'\)) if they induce two disconnected subgraphs of \(\mathcal G\). Otherwise, they are \emph{incompatible} (\(\gamma\not\sim\gamma'\)).
\end{definition}

Any subset $D \subset B$ can be uniquely decomposed into a collection of mutually compatible polymers $\Gamma = \{\gamma_1, \gamma_2, \ldots, \gamma_k\}$ such that $D = \cup_{i=1}^k \gamma_i$. Therefore, we can rewrite $\tilde{P}(\mathbf{x}_{AC})$ as follows:

\begin{equation}\label{eq:polymer_expansion_P}
    \tilde{P}(\mathbf{x}_{AC}) = Z_0 \sum_{\Gamma \text{ compatible}} \prod_{\gamma \in \Gamma} Z_\gamma
\end{equation}
Where the sum runs over all collections of mutually compatible polymers $\Gamma$. $Z_\gamma$ is defined as $Z_D$ but with $D$ replaced by $\gamma$.

Eq.~(\ref{eq:polymer_expansion_P}) in its current form does not converge beacuse the number of compatible polymer collections grows combinatorially. To see that, consider the polymer set $\Gamma$ that contains $k$ single-site polymers. The number of such sets grows as $\binom{|B|}{k}$, which overwhelms the exponential decay of $Z_\gamma$.

To derive a convergent expansion, an important observation is that partition functions is not a stable object in the following sense: adding or removing a single site changes the partition function by a multiplicative factor. On the other hand, free energy $F = \log(Z)$ is a stable object: adding or removing a single site changes the free energy by an additive factor. Therefore, we will write down a polymer expansion for $\log( \tilde{P}(\mathbf{x}_{AC}))$ instead of $\tilde{P}(\mathbf{x}_{AC})$. 

We first define the notion of clusters to simplify notations.
We can also write it as cluster expansion.
\begin{definition}[Clusters]
    A \emph{cluster} $\mathbf{W}$ is a set of tuples $\{(\gamma_1, \mu_1), (\gamma_2, \mu_2), \ldots, (\gamma_m, \mu_m)\}$ where $\gamma_i$ are distinct polymers and $\mu_i$ are positive integers that denote the multiplicity of each polymer. We also Define
    \begin{enumerate}
        \item $|\mathbf{W}| = \sum_{i=1}^m \mu_i |\gamma_i|$ as the cluster weight.
        \item $\mu_\mathbf{W} = \sum_{i=1}^m \mu_i$
        \item $\lambda_{\mathbf{W},L} = \prod_{i=1}^m (\lambda_{\gamma_i,L})^{\mu_i}$
        \item $Z_{\mathbf{W}} = \prod_{i=1}^m (Z_{\gamma_i})^{\mu_i}$
        \item $\mathbf{W}! = \prod_{i=1}^m \mu_i!$
    \end{enumerate}
\end{definition}

Cluster also carries a notion of connectedness. We will define the interaction graph of a cluster to formalize this notion.

\begin{definition}[Cluster Interaction Graph and Connectedness]
    Given a cluster $\mathbf{W} = \{(\gamma_1, \mu_1), (\gamma_2, \mu_2), \ldots, (\gamma_m, \mu_m)\}$, we define the interaction graph $G_{\mathbf{W}}$ as follows: each vertex in $G_{\mathbf{W}}$ corresponds to a polymer $\gamma_i$ in $\mathbf{W}$. There are $\mu_i$ vertices corresponding to $\gamma_i$. Two vertices are connected by an edge if the corresponding polymers are compatible. The cluster $\mathbf{W}$ is said to be connected if $G_{\mathbf{W}}$ is connected.
\end{definition}

With this definition in hand, we quote the expansion for $\log( \tilde{P}(\mathbf{x}_{AC}))$ from standard references below (See Chapter 5 of \cite{friedli2017statistical}).

\begin{lemma}
\label{lem:polymer_expansion_free_energy}
Eq.~(\ref{eq:polymer_expansion_P}) can be reorganized into the following expansion for $\log( \tilde{P}(\mathbf{x}_{AC}))$:
\begin{equation}
    \log( \tilde{P}(\mathbf{x}_{AC})) = \log(Z_0) + \sum_{\mathbf{W} \text{ connected}} \frac{\phi(G_{\mathbf{W}})}{\mathbf{W}!} Z_{\mathbf{W}}
\end{equation}
Where the sum runs over all connected clusters $\mathbf{W}$. $\phi(G_{\mathbf{W}})$ is the Ursell function of the interaction graph $G_{\mathbf{W}}$ defined as follows. When $\mu_\mathbf{W} = 1$, $\phi(G_{\mathbf{W}}) = 1$. When $\mu_\mathbf{W} \geq 2$,
\begin{equation}
    \phi(G_{\mathbf{W}}) = \sum_{C \subseteq G_{\mathbf{W}} \text{ spanning tree}} (-1)^{|E(C)|}
\end{equation}
Where the sum runs over all spanning trees $C$ of $G_{\mathbf{W}}$ and $|E(C)|$ is the number of edges in the spanning tree.
\end{lemma}

After writing down the cluster expansion of $\log( \tilde{P}(\mathbf{x}_{AC}))$, we exploit the locality of connected clusters to decompose $\log( \tilde{P}(\mathbf{x}_{AC}))$ into three parts: a part that only depends on $\mathbf{x}_A$, a part that only depends on $\mathbf{x}_C$, and a part that depends on both $\mathbf{x}_A$ and $\mathbf{x}_C$.
\begin{lemma}
\label{lem:free_energy_decomposition}
The cluster expansion of $\log( \tilde{P}(\mathbf{x}_{AC}))$ can be decomposed as follows:
\begin{equation}
    \log( \tilde{P}(\mathbf{x}_{AC})) = \log(Z_{0,A}) + \log(Z_{0,C}) + F_{\emptyset} + F_A(\mathbf{x}_A) + F_C(\mathbf{x}_C) + F_{AC}(\mathbf{x}_{AC})
\end{equation}
Where $F_{\emptyset}$ is independent of $\mathbf{x}_{AC}$, $F_A(\mathbf{x}_A)$ only depends on $\mathbf{x}_A$, $F_C(\mathbf{x}_C)$ only depends on $\mathbf{x}_C$, and $F_{AC}(\mathbf{x}_{AC})$ depends on both $\mathbf{x}_A$ and $\mathbf{x}_C$. They are defined as follows:
\begin{align}
    F_{\emptyset} &= \sum_{\substack{\mathbf{W} \text{ connected} \\ \mathbf{W} \sim AC}} \frac{\phi(G_{\mathbf{W}})}{\mathbf{W}!} Z_{\mathbf{W}} \\
    F_A(\mathbf{x}_A) &= \sum_{\substack{\mathbf{W} \text{ connected} \\ \mathbf{W} \sim C}} \frac{\phi(G_{\mathbf{W}})}{\mathbf{W}!} Z_{\mathbf{W}} \\
    F_C(\mathbf{x}_C) &= \sum_{\substack{\mathbf{W} \text{ connected} \\ \mathbf{W} \sim A}} \frac{\phi(G_{\mathbf{W}})}{\mathbf{W}!} Z_{\mathbf{W}} \\
    F_{AC}(\mathbf{x}_{AC}) &= \sum_{\substack{\mathbf{W} \text{ connected} \\ \mathbf{W} \not\sim A,C}} \frac{\phi(G_{\mathbf{W}})}{\mathbf{W}!} Z_{\mathbf{W}}
\end{align}
Where $\mathbf{W} \sim A$ denotes clusters in which all polymers in $\mathbf{W}$ are compatible with $A$. $\mathbf{W} \sim C$ and $\mathbf{W} \sim AC$ are similarly defined. $\mathbf{W} \not\sim A,C$ denotes clusters in which at least one polymer in $\mathbf{W}$ is incompatible with $A$ and at least one polymer in $\mathbf{W}$ is incompatible with $C$.

Moreover, let $d_{AC}$ be the distance between $A$ and $C$ in $\mathcal{G}$. Then, the minimal cluster weight $|\mathbf{W}|$ in the expansion of $F_{AC}(\mathbf{x}_{AC})$ is at least $d_{AC}$.
\end{lemma}
\begin{proof}
    The factorization of $Z_0$ follows from Proposition \ref{prop:Z_0_factorization}. Next, we note that any connected cluster $\mathbf{W}$ must fall into one of the following four categories:
    \begin{enumerate}
        \item All polymers in $\mathbf{W}$ are compatible with $AC$.
        \item At least one polymer in $\mathbf{W}$ is incompatible with $A$ but all polymers in $\mathbf{W}$ are compatible with $C$.
        \item At least one polymer in $\mathbf{W}$ is incompatible with $C$ but all polymers in $\mathbf{W}$ are compatible with $A$.
        \item At least one polymer in $\mathbf{W}$ is incompatible with $A$ and at least one polymer in $\mathbf{W}$ is incompatible with $C$.
    \end{enumerate}
    Therefore, we can decompose the cluster expansion of $\log( \tilde{P}(\mathbf{x}_{AC}))$ into four parts according to the above three categories. Next, observe that for a conected cluster $\mathbf{W}$ in the third category, there must exist a path of polymers connecting $A$ and $C$ in the interaction graph $G_{\mathbf{W}}$. Since each polymer is connected in $\mathcal{G}$, one can always create a path from spins contained in $\mathbf{W}$ that connects $A$ and $C$. Therefore, the minimal cluster weight $|\mathbf{W}|$ in the expansion of $F_{AC}(\mathbf{x}_{AC})$ is at least $d_{AC}$.
\end{proof}

Lemma \ref{lem:free_energy_decomposition} establishes the desired factorization in Eq.~(\ref{eq:tilde_P_factorization}) with $Z_A(\mathbf{x}_A) = Z_{0,A} F_{\emptyset} \exp(F_A(\mathbf{x}_A))$, $Z_C(\mathbf{x}_C) = Z_{0,C} \exp(F_C(\mathbf{x}_C))$, and $F_{AC}(\mathbf{x}_{AC})$ defined as above (again we absorb $F_{\emptyset}$ into $Z_A$ without loss of generality). The last step is to show that $|F_{AC}(\mathbf{x}_{AC})|$ decays exponentially with $d_{AC}$ when $p_{\min}$ is above a critical threshold. This follows from the convergence of the polymer expansion, which we prove in the next section.

\subsection{Convergence of the Polymer Expansion}
We now prove that the polymer expansion for $\log( \tilde{P}(\mathbf{x}_{AC}))$ converges exponentially fast when $p_{\min}$ is above a critical threshold. This establishes the desired decay of $|F_{AC}(\mathbf{x}_{AC})|$ with $d_{AC}$. We use the Kotecky-Preiss condition which gives a blackox criterion for the convergence of polymer expansions.

\begin{lemma}[Kotecký–Preiss criterion for the cluster expansion \cite{kotecky1986cluster}]\label{lem:KP}
  If there exists two constants $a$, $b$ such that for every polymer $\gamma$ we have
\begin{equation}\label{eq:KPcondition}
  \sum_{\gamma':\;\gamma'\not\sim \gamma} |Z_{\gamma'}|\,e^{(a+b)|\gamma'|} \le a|\gamma|,
\end{equation}
then the expansion in Lemma \ref{lem:polymer_expansion_free_energy} converges absolutely. Moreover, for any spin $i$, we have the bound on the convergence:
\begin{equation}\label{eq:KPbound}
  \sum_{\substack{
    \text{connected }\, \mathbf{W} \not\sim i
  }} \left|\frac{\phi(G_{\mathbf{W}})}{\mathbf{W}!} \lambda_{\mathbf{W},L}\right| e^{b|\mathbf{W}|} \le a
\end{equation}
Where $\mathbf{W} \not\sim i$ denotes clusters in which at least one polymer in $\mathbf{W}$ contains site $i$.
\end{lemma}

To apply the Kotecký–Preiss criterion, we need to bound the left-hand side of Eq.~(\ref{eq:KPcondition}). This is done using Lemma \ref{lem:Z_D_bound}, which shows that $|Z_\gamma|$ decays exponentially with $|\gamma|$, and a combinatorial bound on the number of polymers of size $k$ that are incompatible with a given polymer $\gamma$.
\begin{lemma}
\label{lem:incompatible_polymers_bound}
For any polymer $\gamma$, the number of polymers $\gamma'$ of size $k$ that are incompatible with $\gamma$ is bounded as follows:
\begin{equation}
    N_{\gamma,k} \leq |\gamma| (e\mathfrak{d})^{k-1}
\end{equation}
\end{lemma}
\begin{proof}
    We first choose a site in $\gamma$ and bound the number of polymers of size $k$ that contain this site. This reduces the problem to counting the number of connected subgraphs of size $k$ that contain a given site, a standard problem known to be $(e \mathfrak{d})^{k-1}$. Since there are $|\gamma|$ choices of the site in $\gamma$, we have the desired bound.
\end{proof}

With Lemma \ref{lem:Z_D_bound} and Lemma \ref{lem:incompatible_polymers_bound}, we can now bound the left-hand side of Eq.~(\ref{eq:KPcondition}) as follows:

\begin{proposition}\label{prop:KP_condition_satisfaction}
    Eq.~(\ref{eq:KPcondition}) can be satisfied when $\delta \geq \delta_c = 1 + \log(1+\mathfrak{d})$. In particular, we can set $a=1$ and any $b \le \delta - \delta_c$. In terms of $p_{\min}$, we need $p_{\min} \geq p_{\min,c} = \mathfrak{d}(\delta_c + \mathfrak{d} h_{\max} + \log(q-1))$. Correspondingly, we can set any $b \leq p_{\min} - p_{\min,c}$.
\end{proposition}
\begin{proof}
    We apply Lemma \ref{lem:Z_D_bound} and Lemma \ref{lem:incompatible_polymers_bound} to bound the left-hand side of Eq.~(\ref{eq:KPcondition}):
\begin{align}
    \sum_{\gamma':\;\gamma'\not\sim \gamma} |Z_{\gamma'}|\,e^{(a+b)|\gamma'|} &= \sum_{k=1}^\infty \sum_{\substack{\gamma':\;\gamma'\not\sim \gamma \\ |\gamma'| = k}} |Z_{\gamma'}|\,e^{(a+b)k} \\
    &\leq \sum_{k=1}^\infty N_{\gamma,k} e^{-\delta k} e^{(a+b)k} \\
    &\leq |\gamma| \sum_{k=1}^\infty (e\mathfrak{d})^{k-1} e^{-\delta k} e^{(a+b)k} \\
\end{align}
As a reminder, $\delta = p_{\min}/\mathfrak{d} - \mathfrak{d} h_{\max} - \log(q-1)$ is controlled by the pinning strength $p_{\min}$. Therefore, when $p_{\min}$ is sufficiently large such that $\delta > a + b + \log(\mathfrak{d})$, the above sum converges to
\begin{equation}
    \sum_{\gamma':\;\gamma'\not\sim \gamma} |Z_{\gamma'}|\,e^{(a+b)|\gamma'|} \leq |\gamma| \frac{e^{a+b-\delta}}{1-e^{a+b-\delta}\mathfrak{d}}
\end{equation}
We will set $a=1$. To satisfy the Kotecký–Preiss condition in Eq.~(\ref{eq:KPcondition}), we need
\begin{equation}
    \frac{e^{1+b-\delta}}{1-e^{1+b-\delta}\mathfrak{d}} \leq 1
\end{equation}
This is achieved when $\delta \geq 1 + b + \log(1+\mathfrak{d})$. Combining with the convergence condition $\delta > 1 + b + \log(\mathfrak{d})$, we have the desired result.
\end{proof}

With the above proposition, we have established the convergence of the polymer expansion when $p_{\min}$ is above a critical threshold. Using the exponential convergence, we can now bound $|F_{AC}(\mathbf{x}_{AC})|$ as follows:
\begin{lemma}
\label{lem:F_AC_bound}
    When $p_{\min} \geq p_{\min,c}$ as defined in Proposition \ref{prop:KP_condition_satisfaction}, we have
    \begin{equation}
        |F_{AC}(\mathbf{x}_{AC})| \leq \min(| \partial A |, | \partial C |) e^{- (p_{\min} - p_{\min,c}) d_{AC}}
    \end{equation}
\end{lemma}
\begin{proof}
    We first apply the decomposition in Lemma \ref{lem:free_energy_decomposition} to write down $F_{AC}(\mathbf{x}_{AC})$ as a sum over connected clusters. with weight at least $d_{AC}$.
    \begin{equation}
        |F_{AC}(\mathbf{x}_{AC})| \leq \sum_{k=d_{AC} }^\infty \sum_{\substack{\mathbf{W} \text{ connected} \\ |\mathbf{W}| = k \\ \mathbf{W} \not\sim A,C}} \left|\frac{\phi(G_{\mathbf{W}})}{\mathbf{W}!} Z_{\mathbf{W}} \right|  \leq \sum_{k=d_{AC} }^\infty \sum_{\substack{\mathbf{W} \text{ connected} \\ |\mathbf{W}| = k}} \left|\frac{\phi(G_{\mathbf{W}})}{\mathbf{W}!} Z_{\mathbf{W}} \right|
    \end{equation}
    Next, we note that any cluster $\mathbf{W}$ in the above sum must contain at least one polymer that is incompatible with $\partial A$ and $\partial C$ because $\mathbf{W}$ must cross both boundaries to connect $A$ and $C$. Therefore, we can bound the above sum as follows:
    \begin{equation}
        |F_{AC}(\mathbf{x}_{AC})| \leq \min(| \partial A |, | \partial C |) \sum_{k=d_{AC} }^\infty \sum_{\substack{\mathbf{W} \text{ connected} \\ |\mathbf{W}| = k \\ \mathbf{W} \not\sim i}} \left|\frac{\phi(G_{\mathbf{W}})}{\mathbf{W}!} Z_{\mathbf{W}} \right|
    \end{equation}
    Where $i$ is any site in $\partial A$ or $\partial C$, and we choose the smaller of the two to sum over sites. Finally, we apply the Kotecký–Preiss bound in Eq.~(\ref{eq:KPbound}) to bound the above sum. We multiply the above equation by $e^{b d_{AC}}$:
    \begin{align}
        &\sum_{k=d_{AC} }^\infty \sum_{\substack{\mathbf{W} \text{ connected} \\ |\mathbf{W}| = k \\ \mathbf{W} \not\sim i}} \left|\frac{\phi(G_{\mathbf{W}})}{\mathbf{W}!} Z_{\mathbf{W}} \right| e^{b d_{AC}} \\
         \leq &\sum_{k=d_{AC} }^\infty \sum_{\substack{\mathbf{W} \text{ connected} \\ |\mathbf{W}| = k \\ \mathbf{W} \not\sim i}} \left|\frac{\phi(G_{\mathbf{W}})}{\mathbf{W}!} Z_{\mathbf{W}} \right| e^{b k} \\
        \leq &\sum_{k=1}^\infty \sum_{\substack{\mathbf{W} \text{ connected} \\ |\mathbf{W}| = k \\ \mathbf{W} \not\sim i}} \left|\frac{\phi(G_{\mathbf{W}})}{\mathbf{W}!} Z_{\mathbf{W}} \right| e^{b k} \leq 1
    \end{align}
    Therefore, we have for any $i$,
    \begin{equation}
        \sum_{k=d_{AC} }^\infty \sum_{\substack{\mathbf{W} \text{ connected} \\ |\mathbf{W}| = k \\ \mathbf{W} \not\sim i}} \left|\frac{\phi(G_{\mathbf{W}})}{\mathbf{W}!} Z_{\mathbf{W}} \right| \le e^{-b d_{AC}}
    \end{equation}
    Setting $b=p_{\min} - p_{\min,c}$ gives the exponential decay of $|F_{AC}(\mathbf{x}_{AC})|$ with $d_{AC}$.
\end{proof}

We are now ready to conclude the proof of our main result Lemma \ref{lem:decay_of_correlation}
\begin{proof}[Proof of Lemma \ref{lem:decay_of_correlation}]
    The desired factorization of $\tilde{P}(\mathbf{x}_{AC})$ follows from Lemma \ref{lem:free_energy_decomposition}. The exponential decay of $|F_{AC}(\mathbf{x}_{AC})|$ with $d_{AC}$ follows from Lemma \ref{lem:F_AC_bound}. The rest of the proof follows from Section \ref{sec:proof_outline}.
\end{proof}

\section{Formally Stating the Main Theorems}\label{constants}
In this section, we formally state and prove all main theorems stated in the main text using Lemma \ref{lem:decay_of_correlation}. We start from classical Gibbs distributions under finite-depth local channels. From Lemma \ref{lem:pinning_finite_depth} we obtain a lower bound on $p_{\min}$ defined in Eq.~\eqref{eq:pinned_hamiltonian}.
\begin{equation}
    p_{\min} = d \log(1-\epsilon) - \log(\epsilon)
\end{equation}

Next, recall from Lemma \ref{lem:decay_of_correlation} that the critical pinning strength is given by
\begin{equation}
    p_{\min,c} = \mathfrak{d}(\delta_c + \mathfrak{d} \beta + \log(q^{d}-1)) \le \mathfrak{d}(1 + \log(1+\mathfrak{d}) + \mathfrak{d} \beta + d\log(q))
\end{equation}
Here we have made two substitutions: $h_{\max} = \beta$ since we are considering classical Gibbs distributions at inverse temperature $\beta$ and we have normalized the maximal value of $h_a$ to be $1$. In addition, the local dimension is $q^{d}$ since we consider the blocked spins $X_i = (x_{i,0}, x_{i,2}, \ldots, x_{i,d-1})$. In the second inequality, we have plugged in the value of $\delta_c = 1 + \log(1+\mathfrak{d})$ and used the bound $\log(q^{d}-1) \leq d\log(q)$.

Therefore, to ensure $p_{\min} \geq p_{\min,c}$, it suffices to require
\begin{equation}\label{eq:critical_epsilon_condition}
    d \log(1-\epsilon) - \log(\epsilon) \geq \mathfrak{d}(1 + \log(1+\mathfrak{d}) + \mathfrak{d} \beta + d\log(q))
\end{equation}
This gives the desired critical value of $\epsilon_c$. We now state Theorem \ref{thm:classical_stability_finite_depth_informal} in a formal manner.
\begin{theorem}\label{thm:classical_stability_finite_depth}
    Consider an interaction graph $\mathcal{G}$ supporting a Gibbs distribution $P(\mathbf{x}) \propto e^{-\beta H(\mathbf{x})}$ where $H(\mathbf{x})$ is defined in Eq.~(\ref{eq:hamiltonian}) and each $h_a$ is diagonal in the computational basis. Consider finite-depth stochastic processes $\mathcal{T}$ defined in Definition \ref{def:finite_depth_stochastic_process} where each $\mathcal{T}_{a,t}$ is subject to Eq.~(\ref{eq:perturbed_channel}).There exist a constant $\epsilon_c$ defined as
    \begin{equation}
        d \log(1-\epsilon_c) - \log(\epsilon_c) = \mathfrak{d}(1 + \log(1+\mathfrak{d}) + \mathfrak{d} \beta + d\log(q))
    \end{equation}
     If $\epsilon < \epsilon_c$, then for any tripartition of the spins into disjoint subsets $A,B,C$ such that $B$ separates $A$ and $C$ in $\mathcal{G}$, we have
    \begin{equation}
        I_{P}(A:C|B) = O\left(\min(|\partial A|, |\partial C|) e^{- \frac{d_{AC}}{\xi}}\right)
    \end{equation}
    Where $d_{AC}$ is the distance between $A$ and $C$ in $\mathcal{G}$, and the Markov length $\xi = O(1/(p_{\min}-p_{\min,c}))$.
\end{theorem}

For the commuting quantum case, $P(\mathbf{s})$ is defined on a different interaction graph $\mathcal{G}_s$ which can be understood as a subgraph of the dual graph of $\mathcal{G}$. The maximal degree of $\mathcal{G}_s$ is upper bounded by $\mathfrak{k}$, which is the number of terms in the Hamiltonian that each spin participates in. Therefore, we need to replace $\mathfrak{d}$ by $\mathfrak{k}$ in the above derivation. We now state Theorem \ref{thm:quantum_stability_finite_depth_informal} in a formal manner.
\begin{theorem}\label{thm:quantum_stability_finite_depth}
    Consider an interaction graph $\mathcal{G}$ supporting a commuting Gibbs state $\rho \propto e^{-\beta H}$ where $H$ is defined in Eq.~(\ref{eq:hamiltonian}) and each $h_a$ is a product of Pauli operators that commute mutually. 
    Consider finite-depth local channel $\mathcal{E}$ defined in Definition \ref{def:finite_depth_channel} where each $\mathcal{E}_{a,t}$ (a) is stabilizer-mixing and (b) satisfies Eq.~(\ref{eq:perturbed_channel}). 
    There exist a constant $\epsilon_c$ defined as
    \begin{equation}
        d \log(1-\epsilon_c) - \log(\epsilon_c) = \mathfrak{k}(1 + \log(1+\mathfrak{k}) + \mathfrak{k} \beta + d\log(q))
    \end{equation}
    If $\epsilon < \epsilon_c$, then we have for any disjoint subsets $A,C$ of spins separated by $B$,
    \begin{equation}
        I_{\rho}(A:C|B) = O\left(\min(|\partial A|, |\partial C|) e^{- \frac{d_{AC}}{\xi}}\right)
    \end{equation}
    Where $d_{AC}$ is the distance between $A$ and $C$ in $\mathcal{G}$, and the Markov length $\xi = O(1/(p_{\min}-p_{\min,c}))$.
\end{theorem}

The proof of the two theorems follows from the main text, Lemma \ref{lem:decay_of_correlation}, and the above derivation of the critical $\epsilon_c$.


\end{widetext}


\bibliography{apssamp}

\begin{thebibliography}{58}%
\makeatletter
\providecommand \@ifxundefined [1]{%
 \@ifx{#1\undefined}
}%
\providecommand \@ifnum [1]{%
 \ifnum #1\expandafter \@firstoftwo
 \else \expandafter \@secondoftwo
 \fi
}%
\providecommand \@ifx [1]{%
 \ifx #1\expandafter \@firstoftwo
 \else \expandafter \@secondoftwo
 \fi
}%
\providecommand \natexlab [1]{#1}%
\providecommand \enquote  [1]{``#1''}%
\providecommand \bibnamefont  [1]{#1}%
\providecommand \bibfnamefont [1]{#1}%
\providecommand \citenamefont [1]{#1}%
\providecommand \href@noop [0]{\@secondoftwo}%
\providecommand \href [0]{\begingroup \@sanitize@url \@href}%
\providecommand \@href[1]{\@@startlink{#1}\@@href}%
\providecommand \@@href[1]{\endgroup#1\@@endlink}%
\providecommand \@sanitize@url [0]{\catcode `\\12\catcode `\$12\catcode `\&12\catcode `\#12\catcode `\^12\catcode `\_12\catcode `\%12\relax}%
\providecommand \@@startlink[1]{}%
\providecommand \@@endlink[0]{}%
\providecommand \url  [0]{\begingroup\@sanitize@url \@url }%
\providecommand \@url [1]{\endgroup\@href {#1}{\urlprefix }}%
\providecommand \urlprefix  [0]{URL }%
\providecommand \Eprint [0]{\href }%
\providecommand \doibase [0]{https://doi.org/}%
\providecommand \selectlanguage [0]{\@gobble}%
\providecommand \bibinfo  [0]{\@secondoftwo}%
\providecommand \bibfield  [0]{\@secondoftwo}%
\providecommand \translation [1]{[#1]}%
\providecommand \BibitemOpen [0]{}%
\providecommand \bibitemStop [0]{}%
\providecommand \bibitemNoStop [0]{.\EOS\space}%
\providecommand \EOS [0]{\spacefactor3000\relax}%
\providecommand \BibitemShut  [1]{\csname bibitem#1\endcsname}%
\let\auto@bib@innerbib\@empty
\bibitem [{\citenamefont {Hastings}\ and\ \citenamefont {Wen}(2005)}]{hastings2005quasiadiabatic}%
  \BibitemOpen
  \bibfield  {author} {\bibinfo {author} {\bibfnamefont {M.~B.}\ \bibnamefont {Hastings}}\ and\ \bibinfo {author} {\bibfnamefont {X.-G.}\ \bibnamefont {Wen}},\ }\bibfield  {title} {\bibinfo {title} {Quasiadiabatic continuation of quantum states: The stability of topological ground-state degeneracy and emergent gauge invariance},\ }\href@noop {} {\bibfield  {journal} {\bibinfo  {journal} {Physical Review B—Condensed Matter and Materials Physics}\ }\textbf {\bibinfo {volume} {72}},\ \bibinfo {pages} {045141} (\bibinfo {year} {2005})}\BibitemShut {NoStop}%
\bibitem [{\citenamefont {Anshu}\ and\ \citenamefont {Nirkhe}(2020)}]{anshu2020circuit}%
  \BibitemOpen
  \bibfield  {author} {\bibinfo {author} {\bibfnamefont {A.}~\bibnamefont {Anshu}}\ and\ \bibinfo {author} {\bibfnamefont {C.}~\bibnamefont {Nirkhe}},\ }\bibfield  {title} {\bibinfo {title} {Circuit lower bounds for low-energy states of quantum code hamiltonians},\ }\href@noop {} {\bibfield  {journal} {\bibinfo  {journal} {arXiv preprint arXiv:2011.02044}\ } (\bibinfo {year} {2020})}\BibitemShut {NoStop}%
\bibitem [{\citenamefont {Parham}(2025)}]{parham2025quantum}%
  \BibitemOpen
  \bibfield  {author} {\bibinfo {author} {\bibfnamefont {N.}~\bibnamefont {Parham}},\ }\bibfield  {title} {\bibinfo {title} {Quantum circuit lower bounds in the magic hierarchy},\ }\href@noop {} {\bibfield  {journal} {\bibinfo  {journal} {arXiv preprint arXiv:2504.19966}\ } (\bibinfo {year} {2025})}\BibitemShut {NoStop}%
\bibitem [{\citenamefont {Rosenthal}(2020)}]{rosenthal2020bounds}%
  \BibitemOpen
  \bibfield  {author} {\bibinfo {author} {\bibfnamefont {G.}~\bibnamefont {Rosenthal}},\ }\bibfield  {title} {\bibinfo {title} {Bounds on the qac0 complexity of approximating parity},\ }\href@noop {} {\bibfield  {journal} {\bibinfo  {journal} {arXiv preprint arXiv:2008.07470}\ } (\bibinfo {year} {2020})}\BibitemShut {NoStop}%
\bibitem [{\citenamefont {Nadimpalli}\ \emph {et~al.}(2024)\citenamefont {Nadimpalli}, \citenamefont {Parham}, \citenamefont {Vasconcelos},\ and\ \citenamefont {Yuen}}]{nadimpalli2024pauli}%
  \BibitemOpen
  \bibfield  {author} {\bibinfo {author} {\bibfnamefont {S.}~\bibnamefont {Nadimpalli}}, \bibinfo {author} {\bibfnamefont {N.}~\bibnamefont {Parham}}, \bibinfo {author} {\bibfnamefont {F.}~\bibnamefont {Vasconcelos}},\ and\ \bibinfo {author} {\bibfnamefont {H.}~\bibnamefont {Yuen}},\ }\bibfield  {title} {\bibinfo {title} {On the pauli spectrum of qac0},\ }in\ \href@noop {} {\emph {\bibinfo {booktitle} {Proceedings of the 56th Annual ACM Symposium on Theory of Computing}}}\ (\bibinfo {year} {2024})\ pp.\ \bibinfo {pages} {1498--1506}\BibitemShut {NoStop}%
\bibitem [{\citenamefont {Tantivasadakarn}\ \emph {et~al.}(2023{\natexlab{a}})\citenamefont {Tantivasadakarn}, \citenamefont {Verresen},\ and\ \citenamefont {Vishwanath}}]{tantivasadakarn2023shortest}%
  \BibitemOpen
  \bibfield  {author} {\bibinfo {author} {\bibfnamefont {N.}~\bibnamefont {Tantivasadakarn}}, \bibinfo {author} {\bibfnamefont {R.}~\bibnamefont {Verresen}},\ and\ \bibinfo {author} {\bibfnamefont {A.}~\bibnamefont {Vishwanath}},\ }\bibfield  {title} {\bibinfo {title} {Shortest route to non-abelian topological order on a quantum processor},\ }\href@noop {} {\bibfield  {journal} {\bibinfo  {journal} {Physical Review Letters}\ }\textbf {\bibinfo {volume} {131}},\ \bibinfo {pages} {060405} (\bibinfo {year} {2023}{\natexlab{a}})}\BibitemShut {NoStop}%
\bibitem [{\citenamefont {Tantivasadakarn}\ \emph {et~al.}(2023{\natexlab{b}})\citenamefont {Tantivasadakarn}, \citenamefont {Vishwanath},\ and\ \citenamefont {Verresen}}]{tantivasadakarn2023hierarchy}%
  \BibitemOpen
  \bibfield  {author} {\bibinfo {author} {\bibfnamefont {N.}~\bibnamefont {Tantivasadakarn}}, \bibinfo {author} {\bibfnamefont {A.}~\bibnamefont {Vishwanath}},\ and\ \bibinfo {author} {\bibfnamefont {R.}~\bibnamefont {Verresen}},\ }\bibfield  {title} {\bibinfo {title} {Hierarchy of topological order from finite-depth unitaries, measurement, and feedforward},\ }\href@noop {} {\bibfield  {journal} {\bibinfo  {journal} {PRX Quantum}\ }\textbf {\bibinfo {volume} {4}},\ \bibinfo {pages} {020339} (\bibinfo {year} {2023}{\natexlab{b}})}\BibitemShut {NoStop}%
\bibitem [{\citenamefont {Bravyi}\ \emph {et~al.}(2022)\citenamefont {Bravyi}, \citenamefont {Kim}, \citenamefont {Kliesch},\ and\ \citenamefont {Koenig}}]{bravyi2022adaptive}%
  \BibitemOpen
  \bibfield  {author} {\bibinfo {author} {\bibfnamefont {S.}~\bibnamefont {Bravyi}}, \bibinfo {author} {\bibfnamefont {I.}~\bibnamefont {Kim}}, \bibinfo {author} {\bibfnamefont {A.}~\bibnamefont {Kliesch}},\ and\ \bibinfo {author} {\bibfnamefont {R.}~\bibnamefont {Koenig}},\ }\bibfield  {title} {\bibinfo {title} {Adaptive constant-depth circuits for manipulating non-abelian anyons},\ }\href@noop {} {\bibfield  {journal} {\bibinfo  {journal} {arXiv preprint arXiv:2205.01933}\ } (\bibinfo {year} {2022})}\BibitemShut {NoStop}%
\bibitem [{\citenamefont {Smith}\ \emph {et~al.}(2023)\citenamefont {Smith}, \citenamefont {Crane}, \citenamefont {Wiebe},\ and\ \citenamefont {Girvin}}]{smith2023deterministic}%
  \BibitemOpen
  \bibfield  {author} {\bibinfo {author} {\bibfnamefont {K.~C.}\ \bibnamefont {Smith}}, \bibinfo {author} {\bibfnamefont {E.}~\bibnamefont {Crane}}, \bibinfo {author} {\bibfnamefont {N.}~\bibnamefont {Wiebe}},\ and\ \bibinfo {author} {\bibfnamefont {S.}~\bibnamefont {Girvin}},\ }\bibfield  {title} {\bibinfo {title} {Deterministic constant-depth preparation of the aklt state on a quantum processor using fusion measurements},\ }\href@noop {} {\bibfield  {journal} {\bibinfo  {journal} {PRX Quantum}\ }\textbf {\bibinfo {volume} {4}},\ \bibinfo {pages} {020315} (\bibinfo {year} {2023})}\BibitemShut {NoStop}%
\bibitem [{\citenamefont {Smith}\ \emph {et~al.}(2024)\citenamefont {Smith}, \citenamefont {Khan}, \citenamefont {Clark}, \citenamefont {Girvin},\ and\ \citenamefont {Wei}}]{smith2024constant}%
  \BibitemOpen
  \bibfield  {author} {\bibinfo {author} {\bibfnamefont {K.~C.}\ \bibnamefont {Smith}}, \bibinfo {author} {\bibfnamefont {A.}~\bibnamefont {Khan}}, \bibinfo {author} {\bibfnamefont {B.~K.}\ \bibnamefont {Clark}}, \bibinfo {author} {\bibfnamefont {S.}~\bibnamefont {Girvin}},\ and\ \bibinfo {author} {\bibfnamefont {T.-C.}\ \bibnamefont {Wei}},\ }\bibfield  {title} {\bibinfo {title} {Constant-depth preparation of matrix product states with adaptive quantum circuits},\ }\href@noop {} {\bibfield  {journal} {\bibinfo  {journal} {arXiv preprint arXiv:2404.16083}\ } (\bibinfo {year} {2024})}\BibitemShut {NoStop}%
\bibitem [{\citenamefont {Sahay}\ and\ \citenamefont {Verresen}(2024{\natexlab{a}})}]{sahay2024classifying}%
  \BibitemOpen
  \bibfield  {author} {\bibinfo {author} {\bibfnamefont {R.}~\bibnamefont {Sahay}}\ and\ \bibinfo {author} {\bibfnamefont {R.}~\bibnamefont {Verresen}},\ }\bibfield  {title} {\bibinfo {title} {Classifying one-dimensional quantum states prepared by a single round of measurements},\ }\href@noop {} {\bibfield  {journal} {\bibinfo  {journal} {arXiv preprint arXiv:2404.16753}\ } (\bibinfo {year} {2024}{\natexlab{a}})}\BibitemShut {NoStop}%
\bibitem [{\citenamefont {Sahay}\ and\ \citenamefont {Verresen}(2024{\natexlab{b}})}]{sahay2024finite}%
  \BibitemOpen
  \bibfield  {author} {\bibinfo {author} {\bibfnamefont {R.}~\bibnamefont {Sahay}}\ and\ \bibinfo {author} {\bibfnamefont {R.}~\bibnamefont {Verresen}},\ }\bibfield  {title} {\bibinfo {title} {Finite-depth preparation of tensor network states from measurement},\ }\href@noop {} {\bibfield  {journal} {\bibinfo  {journal} {arXiv preprint arXiv:2404.17087}\ } (\bibinfo {year} {2024}{\natexlab{b}})}\BibitemShut {NoStop}%
\bibitem [{\citenamefont {Stephen}\ and\ \citenamefont {Hart}(2024)}]{stephen2024preparing}%
  \BibitemOpen
  \bibfield  {author} {\bibinfo {author} {\bibfnamefont {D.~T.}\ \bibnamefont {Stephen}}\ and\ \bibinfo {author} {\bibfnamefont {O.}~\bibnamefont {Hart}},\ }\bibfield  {title} {\bibinfo {title} {Preparing matrix product states via fusion: constraints and extensions},\ }\href@noop {} {\bibfield  {journal} {\bibinfo  {journal} {arXiv preprint arXiv:2404.16360}\ } (\bibinfo {year} {2024})}\BibitemShut {NoStop}%
\bibitem [{\citenamefont {Zhang}\ \emph {et~al.}(2024)\citenamefont {Zhang}, \citenamefont {Gopalakrishnan},\ and\ \citenamefont {Styliaris}}]{zhang2024characterizing}%
  \BibitemOpen
  \bibfield  {author} {\bibinfo {author} {\bibfnamefont {Y.}~\bibnamefont {Zhang}}, \bibinfo {author} {\bibfnamefont {S.}~\bibnamefont {Gopalakrishnan}},\ and\ \bibinfo {author} {\bibfnamefont {G.}~\bibnamefont {Styliaris}},\ }\bibfield  {title} {\bibinfo {title} {Characterizing matrix-product states and projected entangled-pair states preparable via measurement and feedback},\ }\href@noop {} {\bibfield  {journal} {\bibinfo  {journal} {PRX Quantum}\ }\textbf {\bibinfo {volume} {5}},\ \bibinfo {pages} {040304} (\bibinfo {year} {2024})}\BibitemShut {NoStop}%
\bibitem [{\citenamefont {Hu}\ \emph {et~al.}(2025)\citenamefont {Hu}, \citenamefont {Liu}, \citenamefont {Zhang},\ and\ \citenamefont {Gao}}]{hu2025local}%
  \BibitemOpen
  \bibfield  {author} {\bibinfo {author} {\bibfnamefont {F.}~\bibnamefont {Hu}}, \bibinfo {author} {\bibfnamefont {G.}~\bibnamefont {Liu}}, \bibinfo {author} {\bibfnamefont {Y.}~\bibnamefont {Zhang}},\ and\ \bibinfo {author} {\bibfnamefont {X.}~\bibnamefont {Gao}},\ }\bibfield  {title} {\bibinfo {title} {Local diffusion models and phases of data distributions},\ }\href@noop {} {\bibfield  {journal} {\bibinfo  {journal} {arXiv preprint arXiv:2508.06614}\ } (\bibinfo {year} {2025})}\BibitemShut {NoStop}%
\bibitem [{\citenamefont {Coser}\ and\ \citenamefont {P{\'e}rez-Garc{\'\i}a}(2019)}]{coser2019classification}%
  \BibitemOpen
  \bibfield  {author} {\bibinfo {author} {\bibfnamefont {A.}~\bibnamefont {Coser}}\ and\ \bibinfo {author} {\bibfnamefont {D.}~\bibnamefont {P{\'e}rez-Garc{\'\i}a}},\ }\bibfield  {title} {\bibinfo {title} {Classification of phases for mixed states via fast dissipative evolution},\ }\href@noop {} {\bibfield  {journal} {\bibinfo  {journal} {Quantum}\ }\textbf {\bibinfo {volume} {3}},\ \bibinfo {pages} {174} (\bibinfo {year} {2019})}\BibitemShut {NoStop}%
\bibitem [{\citenamefont {Bravyi}\ \emph {et~al.}(2006)\citenamefont {Bravyi}, \citenamefont {Hastings},\ and\ \citenamefont {Verstraete}}]{bravyi2006lieb}%
  \BibitemOpen
  \bibfield  {author} {\bibinfo {author} {\bibfnamefont {S.}~\bibnamefont {Bravyi}}, \bibinfo {author} {\bibfnamefont {M.~B.}\ \bibnamefont {Hastings}},\ and\ \bibinfo {author} {\bibfnamefont {F.}~\bibnamefont {Verstraete}},\ }\bibfield  {title} {\bibinfo {title} {Lieb-robinson bounds and the generation of correlations and topological quantum order},\ }\href@noop {} {\bibfield  {journal} {\bibinfo  {journal} {Physical review letters}\ }\textbf {\bibinfo {volume} {97}},\ \bibinfo {pages} {050401} (\bibinfo {year} {2006})}\BibitemShut {NoStop}%
\bibitem [{\citenamefont {Sang}\ \emph {et~al.}(2023)\citenamefont {Sang}, \citenamefont {Li}, \citenamefont {Hsieh},\ and\ \citenamefont {Yoshida}}]{PRXQuantum.4.040332}%
  \BibitemOpen
  \bibfield  {author} {\bibinfo {author} {\bibfnamefont {S.}~\bibnamefont {Sang}}, \bibinfo {author} {\bibfnamefont {Z.}~\bibnamefont {Li}}, \bibinfo {author} {\bibfnamefont {T.~H.}\ \bibnamefont {Hsieh}},\ and\ \bibinfo {author} {\bibfnamefont {B.}~\bibnamefont {Yoshida}},\ }\bibfield  {title} {\bibinfo {title} {Ultrafast entanglement dynamics in monitored quantum circuits},\ }\href {https://doi.org/10.1103/PRXQuantum.4.040332} {\bibfield  {journal} {\bibinfo  {journal} {PRX Quantum}\ }\textbf {\bibinfo {volume} {4}},\ \bibinfo {pages} {040332} (\bibinfo {year} {2023})}\BibitemShut {NoStop}%
\bibitem [{\citenamefont {Zhang}\ and\ \citenamefont {Gopalakrishnan}(2024)}]{zhang2024nonlocal}%
  \BibitemOpen
  \bibfield  {author} {\bibinfo {author} {\bibfnamefont {Y.}~\bibnamefont {Zhang}}\ and\ \bibinfo {author} {\bibfnamefont {S.}~\bibnamefont {Gopalakrishnan}},\ }\bibfield  {title} {\bibinfo {title} {Nonlocal growth of quantum conditional mutual information under decoherence},\ }\href@noop {} {\bibfield  {journal} {\bibinfo  {journal} {Physical Review A}\ }\textbf {\bibinfo {volume} {110}},\ \bibinfo {pages} {032426} (\bibinfo {year} {2024})}\BibitemShut {NoStop}%
\bibitem [{\citenamefont {Lee}\ \emph {et~al.}(2024)\citenamefont {Lee}, \citenamefont {Oh}, \citenamefont {Wong}, \citenamefont {Chen},\ and\ \citenamefont {Jiang}}]{lee2024universal}%
  \BibitemOpen
  \bibfield  {author} {\bibinfo {author} {\bibfnamefont {S.-u.}\ \bibnamefont {Lee}}, \bibinfo {author} {\bibfnamefont {C.}~\bibnamefont {Oh}}, \bibinfo {author} {\bibfnamefont {Y.}~\bibnamefont {Wong}}, \bibinfo {author} {\bibfnamefont {S.}~\bibnamefont {Chen}},\ and\ \bibinfo {author} {\bibfnamefont {L.}~\bibnamefont {Jiang}},\ }\bibfield  {title} {\bibinfo {title} {Universal spreading of conditional mutual information in noisy random circuits},\ }\href@noop {} {\bibfield  {journal} {\bibinfo  {journal} {Physical Review Letters}\ }\textbf {\bibinfo {volume} {133}},\ \bibinfo {pages} {200402} (\bibinfo {year} {2024})}\BibitemShut {NoStop}%
\bibitem [{\citenamefont {Sang}\ and\ \citenamefont {Hsieh}(2024)}]{sang2024stability}%
  \BibitemOpen
  \bibfield  {author} {\bibinfo {author} {\bibfnamefont {S.}~\bibnamefont {Sang}}\ and\ \bibinfo {author} {\bibfnamefont {T.~H.}\ \bibnamefont {Hsieh}},\ }\bibfield  {title} {\bibinfo {title} {Stability of mixed-state quantum phases via finite markov length},\ }\href@noop {} {\bibfield  {journal} {\bibinfo  {journal} {arXiv preprint arXiv:2404.07251}\ } (\bibinfo {year} {2024})}\BibitemShut {NoStop}%
\bibitem [{\citenamefont {Brown}\ and\ \citenamefont {Poulin}(2012)}]{brown2012quantum}%
  \BibitemOpen
  \bibfield  {author} {\bibinfo {author} {\bibfnamefont {W.}~\bibnamefont {Brown}}\ and\ \bibinfo {author} {\bibfnamefont {D.}~\bibnamefont {Poulin}},\ }\bibfield  {title} {\bibinfo {title} {Quantum markov networks and commuting hamiltonians},\ }\href@noop {} {\bibfield  {journal} {\bibinfo  {journal} {arXiv preprint arXiv:1206.0755}\ } (\bibinfo {year} {2012})}\BibitemShut {NoStop}%
\bibitem [{\citenamefont {Lloyd}\ \emph {et~al.}(2025)\citenamefont {Lloyd}, \citenamefont {Abanin},\ and\ \citenamefont {Gopalakrishnan}}]{lloyd2025diverging}%
  \BibitemOpen
  \bibfield  {author} {\bibinfo {author} {\bibfnamefont {J.}~\bibnamefont {Lloyd}}, \bibinfo {author} {\bibfnamefont {D.~A.}\ \bibnamefont {Abanin}},\ and\ \bibinfo {author} {\bibfnamefont {S.}~\bibnamefont {Gopalakrishnan}},\ }\bibfield  {title} {\bibinfo {title} {Diverging conditional correlation lengths in the approach to high temperature},\ }\href@noop {} {\bibfield  {journal} {\bibinfo  {journal} {arXiv preprint arXiv:2508.02567}\ } (\bibinfo {year} {2025})}\BibitemShut {NoStop}%
\bibitem [{\citenamefont {Sang}\ \emph {et~al.}(2025)\citenamefont {Sang}, \citenamefont {Lessa}, \citenamefont {Mong}, \citenamefont {Grover}, \citenamefont {Wang},\ and\ \citenamefont {Hsieh}}]{sang2025mixed}%
  \BibitemOpen
  \bibfield  {author} {\bibinfo {author} {\bibfnamefont {S.}~\bibnamefont {Sang}}, \bibinfo {author} {\bibfnamefont {L.~A.}\ \bibnamefont {Lessa}}, \bibinfo {author} {\bibfnamefont {R.~S.}\ \bibnamefont {Mong}}, \bibinfo {author} {\bibfnamefont {T.}~\bibnamefont {Grover}}, \bibinfo {author} {\bibfnamefont {C.}~\bibnamefont {Wang}},\ and\ \bibinfo {author} {\bibfnamefont {T.~H.}\ \bibnamefont {Hsieh}},\ }\bibfield  {title} {\bibinfo {title} {Mixed-state phases from local reversibility},\ }\href@noop {} {\bibfield  {journal} {\bibinfo  {journal} {arXiv preprint arXiv:2507.02292}\ } (\bibinfo {year} {2025})}\BibitemShut {NoStop}%
\bibitem [{\citenamefont {Zhang}\ and\ \citenamefont {Gopalakrishnan}(2025)}]{zhang2025conditional}%
  \BibitemOpen
  \bibfield  {author} {\bibinfo {author} {\bibfnamefont {Y.~F.}\ \bibnamefont {Zhang}}\ and\ \bibinfo {author} {\bibfnamefont {S.}~\bibnamefont {Gopalakrishnan}},\ }\bibfield  {title} {\bibinfo {title} {Conditional mutual information and information-theoretic phases of decohered gibbs states},\ }\href@noop {} {\bibfield  {journal} {\bibinfo  {journal} {Physical Review Letters}\ }\textbf {\bibinfo {volume} {135}},\ \bibinfo {pages} {160401} (\bibinfo {year} {2025})}\BibitemShut {NoStop}%
\bibitem [{Note1()}]{Note1}%
  \BibitemOpen
  \bibinfo {note} {Strictly speaking, Ref.~\cite {zhang2025conditional} only shows the finite Markov length under certain technical restrictions on the channel. These properties resemble the stabilzier mixing condition in this paper in the sense that it preserves the commutation property of the algebra genearting the Gibbs state. However, we believe this is a technical constraint; these results should hold for any high-temperature (non-commuting) Gibbs states under arbitrary strictly local channels.}\BibitemShut {Stop}%
\bibitem [{\citenamefont {Ma}\ \emph {et~al.}(2025)\citenamefont {Ma}, \citenamefont {Khemani},\ and\ \citenamefont {Sang}}]{ma2025circuit}%
  \BibitemOpen
  \bibfield  {author} {\bibinfo {author} {\bibfnamefont {R.}~\bibnamefont {Ma}}, \bibinfo {author} {\bibfnamefont {V.}~\bibnamefont {Khemani}},\ and\ \bibinfo {author} {\bibfnamefont {S.}~\bibnamefont {Sang}},\ }\bibfield  {title} {\bibinfo {title} {Circuit-based characterization of finite-temperature quantum phases and self-correcting quantum memory},\ }\href@noop {} {\bibfield  {journal} {\bibinfo  {journal} {arXiv preprint arXiv:2509.15204}\ } (\bibinfo {year} {2025})}\BibitemShut {NoStop}%
\bibitem [{\citenamefont {Yi}\ \emph {et~al.}(2025)\citenamefont {Yi}, \citenamefont {Li}, \citenamefont {Liu}, \citenamefont {Li},\ and\ \citenamefont {Zou}}]{yi2025universal}%
  \BibitemOpen
  \bibfield  {author} {\bibinfo {author} {\bibfnamefont {J.}~\bibnamefont {Yi}}, \bibinfo {author} {\bibfnamefont {K.}~\bibnamefont {Li}}, \bibinfo {author} {\bibfnamefont {C.}~\bibnamefont {Liu}}, \bibinfo {author} {\bibfnamefont {Z.}~\bibnamefont {Li}},\ and\ \bibinfo {author} {\bibfnamefont {L.}~\bibnamefont {Zou}},\ }\bibfield  {title} {\bibinfo {title} {Universal decay of (conditional) mutual information in gapped pure- and mixed-state quantum matter},\ }\href {https://arxiv.org/abs/2510.22867} {\bibfield  {journal} {\bibinfo  {journal} {arXiv preprint arXiv:2510.22867}\ } (\bibinfo {year} {2025})}\BibitemShut {NoStop}%
\bibitem [{Note2()}]{Note2}%
  \BibitemOpen
  \bibinfo {note} {Technically speaking, their result only shows that if the ground states of one Hamiltonian has CMI decaying superpolynomially fast, then the ground states of all Hamiltonians in this family have CMI decaying superpolynomially fast. This is slightly weaker than a expoentially decay. The relaxation to superpolynomial decay also seems necessary because of the nature of quasi-adiabatic continuation. However,}\BibitemShut {NoStop}%
\bibitem [{\citenamefont {Lieb}\ and\ \citenamefont {Ruskai}(1973)}]{lieb1973proof}%
  \BibitemOpen
  \bibfield  {author} {\bibinfo {author} {\bibfnamefont {E.~H.}\ \bibnamefont {Lieb}}\ and\ \bibinfo {author} {\bibfnamefont {M.~B.}\ \bibnamefont {Ruskai}},\ }\bibfield  {title} {\bibinfo {title} {Proof of the strong subadditivity of quantum-mechanical entropy},\ }\href@noop {} {\bibfield  {journal} {\bibinfo  {journal} {Les rencontres physiciens-math{\'e}maticiens de Strasbourg-RCP25}\ }\textbf {\bibinfo {volume} {19}},\ \bibinfo {pages} {36} (\bibinfo {year} {1973})}\BibitemShut {NoStop}%
\bibitem [{Note3()}]{Note3}%
  \BibitemOpen
  \bibinfo {note} {The rationale for this terminology is that if we think of $A,B,C$ as respectively the past, present, and future, in a Markov chain the past affects the future only through its effect on the present.}\BibitemShut {Stop}%
\bibitem [{\citenamefont {Hayden}\ \emph {et~al.}(2004)\citenamefont {Hayden}, \citenamefont {Jozsa}, \citenamefont {Petz},\ and\ \citenamefont {Winter}}]{hayden2004structure}%
  \BibitemOpen
  \bibfield  {author} {\bibinfo {author} {\bibfnamefont {P.}~\bibnamefont {Hayden}}, \bibinfo {author} {\bibfnamefont {R.}~\bibnamefont {Jozsa}}, \bibinfo {author} {\bibfnamefont {D.}~\bibnamefont {Petz}},\ and\ \bibinfo {author} {\bibfnamefont {A.}~\bibnamefont {Winter}},\ }\bibfield  {title} {\bibinfo {title} {Structure of states which satisfy strong subadditivity of quantum entropy with equality},\ }\href@noop {} {\bibfield  {journal} {\bibinfo  {journal} {Communications in mathematical physics}\ }\textbf {\bibinfo {volume} {246}},\ \bibinfo {pages} {359} (\bibinfo {year} {2004})}\BibitemShut {NoStop}%
\bibitem [{\citenamefont {Fawzi}\ and\ \citenamefont {Renner}(2015)}]{fawzi2015quantum}%
  \BibitemOpen
  \bibfield  {author} {\bibinfo {author} {\bibfnamefont {O.}~\bibnamefont {Fawzi}}\ and\ \bibinfo {author} {\bibfnamefont {R.}~\bibnamefont {Renner}},\ }\bibfield  {title} {\bibinfo {title} {Quantum conditional mutual information and approximate markov chains},\ }\href@noop {} {\bibfield  {journal} {\bibinfo  {journal} {Communications in Mathematical Physics}\ }\textbf {\bibinfo {volume} {340}},\ \bibinfo {pages} {575} (\bibinfo {year} {2015})}\BibitemShut {NoStop}%
\bibitem [{\citenamefont {Junge}\ \emph {et~al.}(2018)\citenamefont {Junge}, \citenamefont {Renner}, \citenamefont {Sutter}, \citenamefont {Wilde},\ and\ \citenamefont {Winter}}]{junge2018universal}%
  \BibitemOpen
  \bibfield  {author} {\bibinfo {author} {\bibfnamefont {M.}~\bibnamefont {Junge}}, \bibinfo {author} {\bibfnamefont {R.}~\bibnamefont {Renner}}, \bibinfo {author} {\bibfnamefont {D.}~\bibnamefont {Sutter}}, \bibinfo {author} {\bibfnamefont {M.~M.}\ \bibnamefont {Wilde}},\ and\ \bibinfo {author} {\bibfnamefont {A.}~\bibnamefont {Winter}},\ }\bibfield  {title} {\bibinfo {title} {Universal recovery maps and approximate sufficiency of quantum relative entropy},\ }in\ \href@noop {} {\emph {\bibinfo {booktitle} {Annales Henri Poincar{\'e}}}},\ Vol.~\bibinfo {volume} {19}\ (\bibinfo {organization} {Springer},\ \bibinfo {year} {2018})\ pp.\ \bibinfo {pages} {2955--2978}\BibitemShut {NoStop}%
\bibitem [{\citenamefont {Kuwahara}(2025)}]{kuwahara2025clustering}%
  \BibitemOpen
  \bibfield  {author} {\bibinfo {author} {\bibfnamefont {T.}~\bibnamefont {Kuwahara}},\ }\bibfield  {title} {\bibinfo {title} {Clustering of conditional mutual information and quantum markov structure at arbitrary temperatures},\ }\href@noop {} {\bibfield  {journal} {\bibinfo  {journal} {Physical Review X}\ }\textbf {\bibinfo {volume} {15}},\ \bibinfo {pages} {041010} (\bibinfo {year} {2025})}\BibitemShut {NoStop}%
\bibitem [{\citenamefont {Kato}\ and\ \citenamefont {Kuwahara}()}]{kato2504clustering}%
  \BibitemOpen
  \bibfield  {author} {\bibinfo {author} {\bibfnamefont {K.}~\bibnamefont {Kato}}\ and\ \bibinfo {author} {\bibfnamefont {T.}~\bibnamefont {Kuwahara}},\ }\bibfield  {title} {\bibinfo {title} {Clustering of conditional mutual information via quantum belief-propagation channels},\ }\href@noop {} {\bibinfo  {journal} {arXiv preprint arXiv:2504.02235}\ ,\ \bibinfo {pages} {2}}\BibitemShut {NoStop}%
\bibitem [{\citenamefont {Chen}\ and\ \citenamefont {Rouz{\'e}}(2025)}]{chen2025quantum}%
  \BibitemOpen
\bibfield  {journal} {  }\bibfield  {author} {\bibinfo {author} {\bibfnamefont {C.-F.}\ \bibnamefont {Chen}}\ and\ \bibinfo {author} {\bibfnamefont {C.}~\bibnamefont {Rouz{\'e}}},\ }\bibfield  {title} {\bibinfo {title} {Quantum gibbs states are locally markovian},\ }\href@noop {} {\bibfield  {journal} {\bibinfo  {journal} {arXiv preprint arXiv:2504.02208}\ } (\bibinfo {year} {2025})}\BibitemShut {NoStop}%
\bibitem [{\citenamefont {Bakshi}\ \emph {et~al.}(2025)\citenamefont {Bakshi}, \citenamefont {Liu}, \citenamefont {Moitra},\ and\ \citenamefont {Tang}}]{bakshi2025dobrushin}%
  \BibitemOpen
  \bibfield  {author} {\bibinfo {author} {\bibfnamefont {A.}~\bibnamefont {Bakshi}}, \bibinfo {author} {\bibfnamefont {A.}~\bibnamefont {Liu}}, \bibinfo {author} {\bibfnamefont {A.}~\bibnamefont {Moitra}},\ and\ \bibinfo {author} {\bibfnamefont {E.}~\bibnamefont {Tang}},\ }\bibfield  {title} {\bibinfo {title} {A dobrushin condition for quantum markov chains: Rapid mixing and conditional mutual information at high temperature},\ }\href@noop {} {\bibfield  {journal} {\bibinfo  {journal} {arXiv preprint arXiv:2510.08542}\ } (\bibinfo {year} {2025})}\BibitemShut {NoStop}%
\bibitem [{\citenamefont {Bluhm}\ \emph {et~al.}(2025)\citenamefont {Bluhm}, \citenamefont {Capel},\ and\ \citenamefont {P{\'e}rez-Hern{\'a}ndez}}]{bluhm2025strong}%
  \BibitemOpen
  \bibfield  {author} {\bibinfo {author} {\bibfnamefont {A.}~\bibnamefont {Bluhm}}, \bibinfo {author} {\bibfnamefont {{\'A}.}~\bibnamefont {Capel}},\ and\ \bibinfo {author} {\bibfnamefont {A.}~\bibnamefont {P{\'e}rez-Hern{\'a}ndez}},\ }\bibfield  {title} {\bibinfo {title} {Strong decay of correlations for gibbs states in any dimension: A. bluhm et al.},\ }\href@noop {} {\bibfield  {journal} {\bibinfo  {journal} {Journal of Statistical Physics}\ }\textbf {\bibinfo {volume} {192}},\ \bibinfo {pages} {134} (\bibinfo {year} {2025})}\BibitemShut {NoStop}%
\bibitem [{Note4()}]{Note4}%
  \BibitemOpen
  \bibinfo {note} {Strictly speaking, the result of Ref.~\cite {chen2025quantum} only shows the exponential decay of the recovery error with $d_{AC}$. When mapping this to the decay of CMI, one incurs an additional prefactor $|C|$ which renders the bound vacuous in the thermodynamic limit. However, recovery error is sufficeint for most purposes.}\BibitemShut {Stop}%
\bibitem [{\citenamefont {Yang}\ \emph {et~al.}(2025)\citenamefont {Yang}, \citenamefont {Shi},\ and\ \citenamefont {Lee}}]{yang2025topological}%
  \BibitemOpen
  \bibfield  {author} {\bibinfo {author} {\bibfnamefont {T.-H.}\ \bibnamefont {Yang}}, \bibinfo {author} {\bibfnamefont {B.}~\bibnamefont {Shi}},\ and\ \bibinfo {author} {\bibfnamefont {J.~Y.}\ \bibnamefont {Lee}},\ }\bibfield  {title} {\bibinfo {title} {Topological mixed states: Axiomatic approaches and phases of matter},\ }\href@noop {} {\bibfield  {journal} {\bibinfo  {journal} {arXiv preprint arXiv:2506.04221}\ } (\bibinfo {year} {2025})}\BibitemShut {NoStop}%
\bibitem [{\citenamefont {Dennis}\ \emph {et~al.}(2002)\citenamefont {Dennis}, \citenamefont {Kitaev}, \citenamefont {Landahl},\ and\ \citenamefont {Preskill}}]{dennis2002topological}%
  \BibitemOpen
  \bibfield  {author} {\bibinfo {author} {\bibfnamefont {E.}~\bibnamefont {Dennis}}, \bibinfo {author} {\bibfnamefont {A.}~\bibnamefont {Kitaev}}, \bibinfo {author} {\bibfnamefont {A.}~\bibnamefont {Landahl}},\ and\ \bibinfo {author} {\bibfnamefont {J.}~\bibnamefont {Preskill}},\ }\bibfield  {title} {\bibinfo {title} {Topological quantum memory},\ }\href@noop {} {\bibfield  {journal} {\bibinfo  {journal} {Journal of Mathematical Physics}\ }\textbf {\bibinfo {volume} {43}},\ \bibinfo {pages} {4452} (\bibinfo {year} {2002})}\BibitemShut {NoStop}%
\bibitem [{\citenamefont {Pastawski}\ \emph {et~al.}(2011)\citenamefont {Pastawski}, \citenamefont {Clemente},\ and\ \citenamefont {Cirac}}]{pastawski2011quantum}%
  \BibitemOpen
  \bibfield  {author} {\bibinfo {author} {\bibfnamefont {F.}~\bibnamefont {Pastawski}}, \bibinfo {author} {\bibfnamefont {L.}~\bibnamefont {Clemente}},\ and\ \bibinfo {author} {\bibfnamefont {J.~I.}\ \bibnamefont {Cirac}},\ }\bibfield  {title} {\bibinfo {title} {Quantum memories based on engineered dissipation},\ }\href@noop {} {\bibfield  {journal} {\bibinfo  {journal} {Physical Review A—Atomic, Molecular, and Optical Physics}\ }\textbf {\bibinfo {volume} {83}},\ \bibinfo {pages} {012304} (\bibinfo {year} {2011})}\BibitemShut {NoStop}%
\bibitem [{\citenamefont {Brown}\ \emph {et~al.}(2016)\citenamefont {Brown}, \citenamefont {Loss}, \citenamefont {Pachos}, \citenamefont {Self},\ and\ \citenamefont {Wootton}}]{brown2016quantum}%
  \BibitemOpen
  \bibfield  {author} {\bibinfo {author} {\bibfnamefont {B.~J.}\ \bibnamefont {Brown}}, \bibinfo {author} {\bibfnamefont {D.}~\bibnamefont {Loss}}, \bibinfo {author} {\bibfnamefont {J.~K.}\ \bibnamefont {Pachos}}, \bibinfo {author} {\bibfnamefont {C.~N.}\ \bibnamefont {Self}},\ and\ \bibinfo {author} {\bibfnamefont {J.~R.}\ \bibnamefont {Wootton}},\ }\bibfield  {title} {\bibinfo {title} {Quantum memories at finite temperature},\ }\href@noop {} {\bibfield  {journal} {\bibinfo  {journal} {Reviews of Modern Physics}\ }\textbf {\bibinfo {volume} {88}},\ \bibinfo {pages} {045005} (\bibinfo {year} {2016})}\BibitemShut {NoStop}%
\bibitem [{\citenamefont {Dobrushin}(1968)}]{dobrushin1968description}%
  \BibitemOpen
  \bibfield  {author} {\bibinfo {author} {\bibfnamefont {R.~L.}\ \bibnamefont {Dobrushin}},\ }\bibfield  {title} {\bibinfo {title} {The description of a random field by means of conditional probabilities and conditions for its regularity},\ }\href@noop {} {\bibfield  {journal} {\bibinfo  {journal} {Theory of Probability and Its Applications}\ }\textbf {\bibinfo {volume} {13}},\ \bibinfo {pages} {197} (\bibinfo {year} {1968})}\BibitemShut {NoStop}%
\bibitem [{\citenamefont {Lanford}\ and\ \citenamefont {Ruelle}(1969)}]{lanford1969observables}%
  \BibitemOpen
  \bibfield  {author} {\bibinfo {author} {\bibfnamefont {O.~E.}\ \bibnamefont {Lanford}}\ and\ \bibinfo {author} {\bibfnamefont {D.}~\bibnamefont {Ruelle}},\ }\bibfield  {title} {\bibinfo {title} {Observables at infinity and states with short range correlations in statistical mechanics},\ }\href@noop {} {\bibfield  {journal} {\bibinfo  {journal} {Communications in Mathematical Physics}\ }\textbf {\bibinfo {volume} {13}},\ \bibinfo {pages} {194} (\bibinfo {year} {1969})}\BibitemShut {NoStop}%
\bibitem [{\citenamefont {Hyv{\"a}rinen}(2005)}]{hyvarinen2005estimation}%
  \BibitemOpen
  \bibfield  {author} {\bibinfo {author} {\bibfnamefont {A.}~\bibnamefont {Hyv{\"a}rinen}},\ }\bibfield  {title} {\bibinfo {title} {Estimation of non-normalized statistical models by score matching},\ }\href {https://jmlr.org/papers/volume6/hyvarinen05a/hyvarinen05a.pdf} {\bibfield  {journal} {\bibinfo  {journal} {Journal of Machine Learning Research}\ }\textbf {\bibinfo {volume} {6}},\ \bibinfo {pages} {695} (\bibinfo {year} {2005})}\BibitemShut {NoStop}%
\bibitem [{\citenamefont {Vincent}(2011)}]{vincent2011connection}%
  \BibitemOpen
  \bibfield  {author} {\bibinfo {author} {\bibfnamefont {P.}~\bibnamefont {Vincent}},\ }\bibfield  {title} {\bibinfo {title} {A connection between score matching and denoising autoencoders},\ }\href {https://doi.org/10.1162/NECO_a_00142} {\bibfield  {journal} {\bibinfo  {journal} {Neural Computation}\ }\textbf {\bibinfo {volume} {23}},\ \bibinfo {pages} {1661} (\bibinfo {year} {2011})}\BibitemShut {NoStop}%
\bibitem [{\citenamefont {Sohl-Dickstein}\ \emph {et~al.}(2015)\citenamefont {Sohl-Dickstein}, \citenamefont {Weiss}, \citenamefont {Maheswaranathan},\ and\ \citenamefont {Ganguli}}]{sohl2015deep}%
  \BibitemOpen
  \bibfield  {author} {\bibinfo {author} {\bibfnamefont {J.}~\bibnamefont {Sohl-Dickstein}}, \bibinfo {author} {\bibfnamefont {E.~A.}\ \bibnamefont {Weiss}}, \bibinfo {author} {\bibfnamefont {N.}~\bibnamefont {Maheswaranathan}},\ and\ \bibinfo {author} {\bibfnamefont {S.}~\bibnamefont {Ganguli}},\ }\bibfield  {title} {\bibinfo {title} {Deep unsupervised learning using nonequilibrium thermodynamics},\ }in\ \href {https://proceedings.mlr.press/v37/sohl-dickstein15.html} {\emph {\bibinfo {booktitle} {Proceedings of the 32nd International Conference on Machine Learning (ICML)}}},\ \bibinfo {series} {Proceedings of Machine Learning Research}, Vol.~\bibinfo {volume} {37}\ (\bibinfo  {publisher} {PMLR},\ \bibinfo {address} {Lille, France},\ \bibinfo {year} {2015})\ pp.\ \bibinfo {pages} {2256--2265}\BibitemShut {NoStop}%
\bibitem [{\citenamefont {Ho}\ \emph {et~al.}(2020)\citenamefont {Ho}, \citenamefont {Jain},\ and\ \citenamefont {Abbeel}}]{ho2020denoising}%
  \BibitemOpen
  \bibfield  {author} {\bibinfo {author} {\bibfnamefont {J.}~\bibnamefont {Ho}}, \bibinfo {author} {\bibfnamefont {A.}~\bibnamefont {Jain}},\ and\ \bibinfo {author} {\bibfnamefont {P.}~\bibnamefont {Abbeel}},\ }\bibfield  {title} {\bibinfo {title} {Denoising diffusion probabilistic models},\ }in\ \href {https://arxiv.org/abs/2006.11239} {\emph {\bibinfo {booktitle} {Advances in Neural Information Processing Systems (NeurIPS) 33}}}\ (\bibinfo  {publisher} {Curran Associates, Inc.},\ \bibinfo {year} {2020})\ pp.\ \bibinfo {pages} {6840--6851}\BibitemShut {NoStop}%
\bibitem [{\citenamefont {Song}\ and\ \citenamefont {Ermon}(2019)}]{song2019generative}%
  \BibitemOpen
  \bibfield  {author} {\bibinfo {author} {\bibfnamefont {Y.}~\bibnamefont {Song}}\ and\ \bibinfo {author} {\bibfnamefont {S.}~\bibnamefont {Ermon}},\ }\bibfield  {title} {\bibinfo {title} {Generative modeling by estimating gradients of the data distribution},\ }in\ \href {https://arxiv.org/abs/1907.05600} {\emph {\bibinfo {booktitle} {Advances in Neural Information Processing Systems (NeurIPS) 32}}}\ (\bibinfo {year} {2019})\ pp.\ \bibinfo {pages} {11895--11907}\BibitemShut {NoStop}%
\bibitem [{\citenamefont {Song}\ and\ \citenamefont {Ermon}(2020)}]{song2020improved}%
  \BibitemOpen
  \bibfield  {author} {\bibinfo {author} {\bibfnamefont {Y.}~\bibnamefont {Song}}\ and\ \bibinfo {author} {\bibfnamefont {S.}~\bibnamefont {Ermon}},\ }\bibfield  {title} {\bibinfo {title} {Improved techniques for training score-based generative models},\ }in\ \href {https://arxiv.org/abs/2006.09011} {\emph {\bibinfo {booktitle} {Advances in Neural Information Processing Systems (NeurIPS) 33}}}\ (\bibinfo {year} {2020})\BibitemShut {NoStop}%
\bibitem [{\citenamefont {Nichol}\ and\ \citenamefont {Dhariwal}(2021)}]{nichol2021improved}%
  \BibitemOpen
  \bibfield  {author} {\bibinfo {author} {\bibfnamefont {A.~Q.}\ \bibnamefont {Nichol}}\ and\ \bibinfo {author} {\bibfnamefont {P.}~\bibnamefont {Dhariwal}},\ }\bibfield  {title} {\bibinfo {title} {Improved denoising diffusion probabilistic models},\ }in\ \href {https://arxiv.org/abs/2102.09672} {\emph {\bibinfo {booktitle} {Proceedings of the 38th International Conference on Machine Learning (ICML) 2021}}},\ \bibinfo {series} {Proceedings of Machine Learning Research}, Vol.\ \bibinfo {volume} {139}\ (\bibinfo  {publisher} {PMLR},\ \bibinfo {year} {2021})\ pp.\ \bibinfo {pages} {8162--8171}\BibitemShut {NoStop}%
\bibitem [{\citenamefont {Koteck{\'y}}\ and\ \citenamefont {Preiss}(1986)}]{kotecky1986cluster}%
  \BibitemOpen
  \bibfield  {author} {\bibinfo {author} {\bibfnamefont {R.}~\bibnamefont {Koteck{\'y}}}\ and\ \bibinfo {author} {\bibfnamefont {D.}~\bibnamefont {Preiss}},\ }\bibfield  {title} {\bibinfo {title} {Cluster expansion for abstract polymer models},\ }\href@noop {} {\bibfield  {journal} {\bibinfo  {journal} {Communications in Mathematical Physics}\ }\textbf {\bibinfo {volume} {103}},\ \bibinfo {pages} {491} (\bibinfo {year} {1986})}\BibitemShut {NoStop}%
\bibitem [{\citenamefont {Dobrushin}(1996)}]{dobrushin1996estimates}%
  \BibitemOpen
  \bibfield  {author} {\bibinfo {author} {\bibfnamefont {R.~L.}\ \bibnamefont {Dobrushin}},\ }\bibfield  {title} {\bibinfo {title} {Estimates of semi-invariants for the ising model at low temperatures},\ }\href@noop {} {\bibfield  {journal} {\bibinfo  {journal} {Translations of the American Mathematical Society-Series 2}\ }\textbf {\bibinfo {volume} {177}},\ \bibinfo {pages} {59} (\bibinfo {year} {1996})}\BibitemShut {NoStop}%
\bibitem [{\citenamefont {Friedli}\ and\ \citenamefont {Velenik}(2017)}]{friedli2017statistical}%
  \BibitemOpen
  \bibfield  {author} {\bibinfo {author} {\bibfnamefont {S.}~\bibnamefont {Friedli}}\ and\ \bibinfo {author} {\bibfnamefont {Y.}~\bibnamefont {Velenik}},\ }\href@noop {} {\emph {\bibinfo {title} {Statistical mechanics of lattice systems: a concrete mathematical introduction}}}\ (\bibinfo  {publisher} {Cambridge University Press},\ \bibinfo {year} {2017})\BibitemShut {NoStop}%
\bibitem [{\citenamefont {Van~Enter}\ \emph {et~al.}(1991)\citenamefont {Van~Enter}, \citenamefont {Fern{\'a}ndez},\ and\ \citenamefont {Sokal}}]{van1991renormalization}%
  \BibitemOpen
  \bibfield  {author} {\bibinfo {author} {\bibfnamefont {A.~C.}\ \bibnamefont {Van~Enter}}, \bibinfo {author} {\bibfnamefont {R.}~\bibnamefont {Fern{\'a}ndez}},\ and\ \bibinfo {author} {\bibfnamefont {A.~D.}\ \bibnamefont {Sokal}},\ }\bibfield  {title} {\bibinfo {title} {Renormalization transformations in the vicinity of first-order phase transitions: What can and cannot go wrong},\ }\href@noop {} {\bibfield  {journal} {\bibinfo  {journal} {Physical review letters}\ }\textbf {\bibinfo {volume} {66}},\ \bibinfo {pages} {3253} (\bibinfo {year} {1991})}\BibitemShut {NoStop}%
\bibitem [{\citenamefont {Van~Enter}\ \emph {et~al.}(1993)\citenamefont {Van~Enter}, \citenamefont {Fern{\'a}ndez},\ and\ \citenamefont {Sokal}}]{van1993regularity}%
  \BibitemOpen
  \bibfield  {author} {\bibinfo {author} {\bibfnamefont {A.~C.}\ \bibnamefont {Van~Enter}}, \bibinfo {author} {\bibfnamefont {R.}~\bibnamefont {Fern{\'a}ndez}},\ and\ \bibinfo {author} {\bibfnamefont {A.~D.}\ \bibnamefont {Sokal}},\ }\bibfield  {title} {\bibinfo {title} {Regularity properties and pathologies of position-space renormalization-group transformations: Scope and limitations of gibbsian theory},\ }\href@noop {} {\bibfield  {journal} {\bibinfo  {journal} {Journal of Statistical Physics}\ }\textbf {\bibinfo {volume} {72}},\ \bibinfo {pages} {879} (\bibinfo {year} {1993})}\BibitemShut {NoStop}%
\end{thebibliography}%

\end{document}